\documentclass[sigplan,screen,authorversion]{acmart}

\setcopyright{acmlicensed}
\acmPrice{15.00}
\acmDOI{10.1145/3453483.3454090}
\acmYear{2021}
\copyrightyear{2021}
\acmSubmissionID{pldi21main-p465-p}
\acmISBN{978-1-4503-8391-2/21/06}
\acmConference[PLDI '21]{Proceedings of the 42nd ACM SIGPLAN International Conference on Programming Language Design and Implementation}{June 20--25, 2021}{Virtual, Canada}
\acmBooktitle{Proceedings of the 42nd ACM SIGPLAN International Conference on Programming Language Design and Implementation (PLDI '21), June 20--25, 2021, Virtual, Canada}

\usepackage[figure,algo2e,linesnumbered,vlined]{algorithm2e}
\usepackage[noend]{algpseudocode}
\usepackage{microtype}
\usepackage{graphicx}
\usepackage{xspace}
\usepackage{url}
\usepackage{subcaption}
\usepackage{appendix}
\usepackage{multicol}
\usepackage{placeins}
\usepackage{listings}
\graphicspath{{./pictures/}}

\usepackage{xcolor}
\def\HiLi{\leavevmode\rlap{\hbox to \hsize{\color{lightgray!50}\leaders\hrule height .8\baselineskip depth .5ex\hfill}}}

\SetAlFnt{\scriptsize}

\usepackage{mathtools}

\DeclarePairedDelimiter\floor{\lfloor}{\rfloor}

\lstdefinestyle{mystyle}{
    basicstyle=\ttfamily\tiny,
    showspaces=false,
    showstringspaces=false,
    showtabs=false,
    tabsize=2
}

\newtheorem{theorem}{Theorem}

\newcommand{\hyaline}{Hyaline}

\usepackage{booktabs}

\begin{document}

\title[Hyaline]{Snapshot-Free, Transparent, and Robust Memory Reclamation for Lock-Free Data Structures}

\author{Ruslan Nikolaev}
\email{rnikola@vt.edu}
\affiliation{%
  \department{Department of Electrical and Computer Engineering}
  \institution{Virginia Tech}
  \city{Blacksburg}
  \state{VA}
  \country{USA}
}

\author{Binoy Ravindran}
\email{binoy@vt.edu}
\affiliation{%
  \department{Department of Electrical and Computer Engineering}
  \institution{Virginia Tech}
  \city{Blacksburg}
  \state{VA}
  \country{USA}
}

\begin{abstract}
We present a family of safe memory reclamation schemes, Hyaline, which are fast, scalable, and transparent to the underlying lock-free data structures. Hyaline is based on reference counting -- considered impractical for memory reclamation in the past due to high overheads. Hyaline uses reference counters only during reclamation, but not while accessing individual objects, which reduces overheads for object accesses. Since with reference counters, an arbitrary thread ends up freeing memory, Hyaline's reclamation workload is (almost) balanced across all threads, unlike most prior reclamation schemes such as epoch-based reclamation (EBR) or hazard pointers (HP). Hyaline often yields (excellent) EBR-grade performance with (good) HP-grade memory efficiency, which is a challenging trade-off with all existing schemes.

Hyaline schemes offer: (i) high \emph{performance}; (ii) good memory \emph{efficiency}; (iii) \emph{robustness}: bounding memory usage even in the presence of stalled threads, a well-known problem with EBR; (iv) \emph{transparency}: supporting virtually unbounded number of threads (or concurrent entities) that can be created and deleted dynamically, and effortlessly join existent workload; (v) \emph{autonomy}: avoiding special OS mechanisms and being non-intrusive to runtime or compiler environments; (vi) \emph{simplicity}: enabling easy integration into unmanaged C/C++ code;
and (vii) \emph{generality}: supporting many data structures.
All existing schemes lack one or more properties.

We have implemented and tested Hyaline on x86(-64), ARM32/64, PowerPC, and MIPS. The general approach requires LL/SC or double-width CAS, while a specialized version also works with single-width CAS. Our  evaluation reveals that Hyaline's throughput is very high --  it steadily outperforms EBR by 10\% in one test and yields \textbf{2x} gains in oversubscribed scenarios. Hyaline's superior memory efficiency is especially evident in read-dominated workloads.

\end{abstract}



\maketitle

\algnewcommand{\algorithmicgoto}{\textbf{goto}}%
\algnewcommand{\Goto}[1]{\algorithmicgoto~\ref{#1}}%
\algdef{SE}[DOWHILE]{Do}{doWhile}{\algorithmicdo}[1]{\algorithmicwhile\ #1}%
\algnewcommand\Not{\textbf{not}\xspace}
\algnewcommand\AndOp{\textbf{and}\xspace}
\algnewcommand\ModOp{\textbf{mod}\xspace}
\algnewcommand\OrOp{\textbf{or}\xspace}

\algnewcommand{\LineComment}[1]{\State \(\triangleright\) #1}

\SetKwIF{If}{ElseIf}{Else}{if (}{)}{else if}{else}{endif}

\SetKwRepeat{Do}{do}{while}
\SetKwProg{Fn}{}{}{}

\section{Introduction}

Modern computer systems increasingly rely on parallelism.
Programming paradigms are also changing accordingly:
the use of scalable non-blocking data structures is preferred
to more traditional lock-based approaches.

Aside from general memory allocation and reclamation problems,
non-blocking data structures also present a number of unique challenges that do not
manifest in lock-based programming. One of the most fundamental problems
for lock-free data structures that use dynamic memory allocation is that
memory objects need to be \emph{safely} deallocated. The problem arises when
one thread wants to deallocate a memory object, but concurrent threads
still have stale pointers and are unaware of ongoing memory deallocation. Garbage collectors avoid this problem
by deferring the deallocation until no thread has pointers to
the deallocated memory object. However, \emph{fully} lock-free
garbage collectors are challenging to design, especially with consistent
and limited overheads.

Moreover, it is often impractical to use garbage collectors in languages that are designed for unmanaged code such as C and C++.
To support concurrent data structures in unmanaged languages, a number of techniques have been developed for safe memory reclamation (or SMR).
Many existing approaches for SMR originate from, or improve upon, epoch-based reclamation (EBR)~\cite{epoch1,epoch2} and hazard pointers (HP)~\cite{hazardPointers}. Unlike garbage collectors, these schemes do
not \emph{automatically} determine when memory becomes free. Instead, such
schemes are predicated on \emph{user-specified retire} statements, which are roughly analogous to \emph{free}, with the only difference being that \emph{retire} does not necessarily deallocate memory right away.

EBR uses a simple API and achieves good performance, but lacks protection against stalled threads. This can prevent timely reclamation, resulting in blocking behavior due to memory exhaustion.
HP does not suffer from this problem, but is harder to use and slower in practice.
Some of the other SMR algorithms~\cite{threadScan,forkScan,Balmau:2016:FRM:2935764.2935790,debra,NBR} rely on special operating system (OS) abstractions, which make them difficult to use in certain settings such as within OS kernels or
platform-independent code.
In general, all SMR schemes
have different trade-offs in terms of API simplicity,
throughput, average memory efficiency, and protection against stalled threads.

Although a number of existing SMR schemes~\cite{epoch1,IntervalBased,hazardEra} achieve excellent throughput, their memory efficiency is limited. An implicit assumption of these algorithms is that all threads get more or less \textit{even} shares of memory objects to reclaim. In most existing SMR schemes, a thread that detaches an object from a data structure must eventually reclaim it. This can cause an unbalanced reclamation workload, especially in read-dominated scenarios, where most threads are reading and only a fraction of them modifies data (see examples in Section~\ref{sec:eval}). When most threads are reading and are therefore not deallocating, the reclamation parallelism is reduced, which degrades memory efficiency. To make matters worse, threads also need to periodically peruse their local lists of not-yet-reclaimed objects to check if an object can be safely reclaimed (as in HP) or
check the status of all threads to advance an epoch (as in EBR). The reclamation workload that is skewed toward the writer threads and consequent delayed reclamation can eventually degrade performance (see Section~\ref{sec:eval}). Such performance degradation becomes even more evident in oversubscribed scenarios where there are more threads than cores available. Note that oversubscription is not that uncommon in practice (e.g., consider Go, Erlang, and proposed C++23 concurrency constructs).

Lock-free reference counting (LFRC)~\cite{refCount,Valois:1995:LLL:224964.224988}, another SMR discipline, enables better parallelism in theory: a thread with the last reference frees an object, which often means that an \textit{arbitrary} thread ends up freeing memory. Unfortunately, LFRC typically performs poorly since every object access, even just for reading, requires memory writes and barriers.

In this paper, we revisit reference counting -- considered impractical for concurrent algorithms in the past -- and design an SMR scheme called \hyaline{} (Sections~\ref{sec:hyaline} and~\ref{sec:algorithms}). The key idea of \hyaline{} is to actively use reference counters only during reclamation, but not while accessing individual objects. This reduces overheads for object accesses, while ensuring that the reclamation workload is balanced across all threads, yielding excellent performance as well as excellent memory efficiency. We establish \hyaline{}'s core properties including reclamation safety, lock-freedom, reclamation cost bounds, and robustness (Section~\ref{sec:correctness}).

Hyaline also has a number of other important properties. Unlike most SMR algorithms, which typically require \textit{globally visible}, private, per-thread state (either in static arrays or in dynamically managed lists), \hyaline{} supports virtually unbounded number of threads using a relatively small (fixed) number of shared \textit{slots}, entities that can be shared by multiple threads. Since \hyaline{}'s reclamation is asynchronous (i.e, any thread can free memory allocated by any thread), threads can immediately be recycled without worrying about the fate of its previously deleted but not-yet-freed objects. These two properties ensure that \hyaline{} is less intrusive to applications, enabling its greater \textit{transparency}. \hyaline{} is also well suited for preemptive environments where the number of threads substantially exceeds the number of cores and can change dynamically such as in OS kernels\footnote{Specifically, this is useful for global data structures within OS kernels that support kernel-mode preemption, e.g., Linux. We also have preliminary results with experimental OS designs~\cite{LibrettOS}.}
and server applications with per-client threads (or fibers).

We have implemented and evaluated \hyaline{} on a variety of architectures including  x86(-64), ARM32/64, PowerPC, and MIPS  (Section~\ref{sec:eval}). Our experimental results reveal that, in a number of cases, \hyaline{} demonstrates both excellent throughput as well as excellent memory efficiency, which is difficult for many past SMR schemes to achieve together. \hyaline{}'s substantial throughput gains are also evident: in the Bonsai Tree benchmark, \hyaline{}'s steady gains over EBR, one of the fastest schemes, are $\approx$10\%. In oversubscribed scenarios, \hyaline{} particularly shines: up to \textbf{2x} throughput gains for high-throughput data structures.

We also present an extension of \hyaline{}, called \hyaline{}-S, to deal with stalled threads (Section~\ref{sec:algorithms}). Similar to EBR, basic \hyaline{}'s memory usage can become unbounded if some threads are stalled. We partially adopt the \emph{birth eras} idea, inspired by similar usage in interval-based reclamation (IBR)~\cite{IntervalBased} and hazard eras (HE)~\cite{hazardEra}, and demonstrate how this idea helps to deal with stalled threads in \hyaline{}-S.

The paper's research contribution is the \hyaline{} algorithm and its variants, which are the first SMR schemes that achieve excellent performance, memory efficiency, and other aforementioned properties over a broad range of workloads including read-dominated and oversubscribed scenarios.

\section{Background}
\label{sec:background}

For greater clarity and completeness, we discuss properties and challenges of the existent memory reclamation schemes.

\subsubsection*{Read-Modify-Write}

Lock-free algorithms typically use read-modify-write (RMW)
operations, which atomically read a memory variable, perform some operation
on it, and write back the result.
Modern CPUs implement RMWs via
compare-and-swap (CAS) or a pair of
load-link (LL)/store-conditional (SC) instructions. For better scalability,
some CPUs~\cite{INTEL}
support specialized fetch-and-add (FAA) and \textit{swap} operations.
Also, x86-64 and ARM64 support operations on two
adjacent CPU words (double-width RMW) via the
{\tt cmpxchg16b}~\cite{INTEL} and {\tt ldaxp}/{\tt stlxp}~\cite{ARM:manual}
instructions, respectively.

\hyaline{} requires either double-width CAS, or ordinary LL/SC (see Appendix~\ref{sec:llsc} regarding LL/SC and PowerPC) since a reference counter needs to be coupled with a pointer in some places.
In rare cases, such as in SPARC~\cite{SPARC}, where neither is supported, reference counters can be squeezed with pointers. (SPARC uses
54-bit virtual addresses; 48-bit cache-line aligned pointers where lower 6 bits are 0s can be squeezed with 16-bit counters.) We also present a specialized \hyaline{}-1 version for single-width CAS.

\subsubsection*{API Model}

We focus on the SMR problem in unmanaged code environments such as C/C++. \hyaline{}'s and \hyaline{}-1's basic programming model is similar to that of EBR~\cite{epoch1} and is defined as follows. Memory objects are allocated
using standard OS-defined means. They additionally incorporate
SMR-related headers and are initialized appropriately. Once memory
objects appear in a lock-free data structure, they must be reclaimed
using a two-step procedure. After deleting a pointer from the data structure,
a memory object needs to be \emph{retired}. A memory object is returned to the OS only after the object becomes unreachable by any other concurrent thread.
All data structure operations must be encapsulated between \textit{enter}
and \textit{leave} calls that trigger the use of SMR.

\subsubsection*{Robustness}

One of the biggest downsides of the simplistic API
model described above is that memory usage becomes unbounded in the
presence of stalled threads. We call an algorithm \emph{robust} if
it bounds memory usage for stalled threads.\footnote{Sometimes, this property is also called ``lock-freedom'' (memory-wise). Since memory is finite,
stalled threads prevent other threads from allocating memory when memory is exhausted.
We use the terminology from~\cite{IntervalBased,Balmau:2016:FRM:2935764.2935790,pageFaultRec} to capture this property more precisely.}

\begin{table*}[htbp]
\caption{Comparison of \hyaline{} with existing SMR approaches.}
\begin{center}
\resizebox{.89\textwidth}{!}{%
	\begin{tabular}{ l l l l l l l }
    \toprule
			\textbf{Scheme} & \textbf{Based on} & \textbf{Performance} & \textbf{Robust} & \textbf{Transparent} & \textbf{Header Size} & \textbf{Usage/API} \\
    \midrule
			LFRC~\cite{refCount,Valois:1995:LLL:224964.224988} & - & Very Slow & Yes & Partially & N/A, but 1 extra & Harder, \\
			& & & & (swap) & word per pointer & intrusive \\
			HP~\cite{hazardPointers} & - & Slow & Yes & No (retire) & 1 word & Harder \\
			EBR~\cite{epoch1,epoch2} & RCU~\cite{Mckenney01read-copyupdate} & Fast & No & No (retire) & 1 word & Very easy \\
			DEBRA+~\cite{debra} & EBR & Fast & Partially & No (OS) & 1 word + desc & Harder \\
			PEBR~\cite{PEBR} & EBR, HP & Medium & Yes & No (retire, OS) & 1 word & Medium \\
			HE~\cite{hazardEra} & EBR, HP & Medium & Yes & No (retire) & 3 words (64 bit) & Harder \\
			IBR (2GE)~\cite{hazardEra} & EBR, HP & Fast & Yes & No (retire) & 3 words (64 bit) & Medium \\
			FreeAccess~\cite{Cohen:2018:DSD:3288538.3276513} & AOA~\cite{Cohen:2015:AMR:2814270.2814298} & Fast & Yes & Partially & GC tracking & Modified \\
			& & & & (swap, GC) & & compiler \\
			\textbf{\hyaline{}} & - & Fast & No & \textbf{Yes} & 3 words & \textbf{Very easy} \\
			\hyaline{}-1 & - & Fast & No & Partially & 3 words & Very easy \\
			\textbf{\hyaline{}-S} & \hyaline{}, & Fast & \textbf{Yes}\footnotemark & \textbf{Yes} & 3 words & \textbf{Medium} \\
			& part. HE/IBR & & & & \\
			\hyaline{}-1S & \hyaline{}-1, & Fast & Yes & Partially & 3 words & Medium \\
			& part. HE/IBR & & & & \\
	\bottomrule

\end{tabular}
}
\end{center}
\label{tbl:comparsmr}
\end{table*}

HP~\cite{hazardPointers} was among the first to recognize this problem and propose a safer API model, which wraps every pointer access.
HP more precisely tracks retired objects using pointers and assigns special indices to accessed objects. HE~\cite{hazardEra} adopted the same API model but proposed to internally record eras (epochs) rather than pointers.
More recently, IBR~\cite{IntervalBased} simplified HE's API by abolishing indices when wrapping pointers,
bringing it closer to the original EBR model.

We added additional robustness guarantees to \hyaline{} and \hyaline{}-1 by extending the API with a pointer wrapper method, \textit{deref}, as in IBR. Our robust schemes, \hyaline{}-S and \hyaline{}-1S,
adopt the birth eras approach from HE and IBR to guard against stalled threads.
The main idea is to mark each allocated object with the global counter, so that stalled threads will only hold older objects. Note
that whereas HE and IBR also use retire eras to identify reclamation intervals, Hyaline relies on reference counting.

\footnotetext{Hyaline-S adaptively changes the number of slots to guarantee robustness.}

\subsubsection*{Reclamation Cost}

Many reclamation schemes have a non-constant reclamation cost.
For example, HP, HE, and IBR need to periodically peruse thread local lists of not-yet-reclaimed objects to check if an object can be safely reclaimed.
\hyaline{} does not need to periodically check thread local lists. We discuss and prove reclamation costs in Section~\ref{sec:correctness}.

\subsubsection*{Transparency}

Most existing SMR schemes maintain special entries throughout thread lifecycles -- e.g., 
static arrays indexed by thread IDs. In practice, threads can be created and deleted
dynamically,
and practical implementations~\cite{CONCURRENCYKIT} maintain lists rather than
arrays with per-thread entries. However, this puts an extra
burden on programmers who have to explicitly
register and unregister threads.
This also breaks seamless integration, as
concurrent data structures cannot be accessed outside
thread contexts -- e.g., signal handlers or OS interrupt contexts. Moreover, unregistration is blocking,
as a thread needs to complete deallocation, which
is impossible until all other threads promise
not to access its locally retired objects.
In \hyaline{}, threads can completely forget about previously retired objects after calling \textit{leave}, as they are already (or will be) taken care for by the remaining threads. (In Section~\ref{sec:vectorized-smr}, \textit{retire}
uses local batches, but they can be immediately finalized by allocating a finite number of dummy nodes.)
We call a scheme \textit{transparent} if it avoids the problems mentioned above.

\subsubsection*{Snapshot-Freedom} To make certain schemes (HP, HE, IBR, etc)
more efficient, when
examining what objects can be deleted safely, a \textit{snapshot}
of global state (i.e., all hazard pointers or eras) is taken.
The per-thread snapshot sidesteps expensive cache misses since it can be
consulted repeatedly for \emph{all} not-yet-freed objects. Per-thread snapshots
are typically preallocated, resulting in extra $O(n^2)$ memory usage, which is
substantial as the number of threads, $n$, grows.
Furthermore, pre-allocated memory needs to be expanded if the number
of threads grows dynamically, which presents additional challenges for \textit{transparency}.
In contrast, EBR consults the global state only once per
each examination and does not take snapshots. All Hyaline schemes are also
snapshot-free.

\subsubsection*{Semantics}

Most robust schemes provide different semantics in handling memory objects
that have never been dereferenced.
Whereas non-robust schemes, such as EBR, can work with the original lock-free linked list~\cite{harrisList}, robust schemes (HP, HE, and IBR)
require a modification~\cite{hazardPointers} that timely retires deleted
list nodes. Our non-robust and robust Hyaline schemes have a similar
distinction.
FreeAccess~\cite{Cohen:2018:DSD:3288538.3276513} -- a recent scheme -- specifically tackles the semantics problem. The scheme is still robust, but falls short on transparently handling the \textit{swap} operation
(can be used for better scalability), and needs compiler modifications.
It also uses a garbage collector which is undesirable when fully transparent
memory management is needed, such as in OS kernels.

\subsubsection*{Memory Overhead (Header Size)}

SMR schemes also differ in extra memory required per each node. For example, EBR, HP, and PEBR store a (thread-local) list pointer per node.
HE and IBR additionally require two 64-bit eras.
All Hyaline variants require 3 words which is equivalent to HE and IBR for 64-bit CPUs and more efficient for 32-bit CPUs.

Although EBR/HP/PEBR's overhead can be fully eliminated by allocating an
intermediate container object when retiring, this causes undesirable
circular allocator dependency. In the same vein, HE, IBR, and Hyaline schemes can reduce the overhead to just 1 word.

\subsubsection*{Summary}
Existing approaches are discussed in detail in Section~\ref{sec:related}.
Table~\ref{tbl:comparsmr} presents a qualitative and quantitative comparison of \hyaline{} with other schemes on metrics including performance, robustness, and transparency.
We also categorize API as hard, medium, etc., similar to the discussion in~\cite{IntervalBased}. Although this categorization is somewhat subjective, we note
that the medium difficulty in robust Hyaline-S and Hyaline-1S implies that
\textit{deref} on pointers can be fully hidden using standard language idioms,
such as smart pointers in C++, and no extra programming
language or OS support is needed. This is not true for schemes that rely on OS
mechanisms. Furthermore, schemes that use HP's API require assigning
indices to reserved objects and annotating where a pointer is used for the last time. These cannot be hidden in smart pointers easily and need to be handled explicitly by a programmer.

\section{\hyaline{}}
\label{sec:hyaline}

\hyaline{} is a member of the family of memory reclamation techniques where programs explicitly retire objects and ensure that retired objects are not reachable from subsequent operations on the data structures. In addition, each operation on the data structures must be enclosed between \textit{enter} and \textit{leave} calls as presented in Figure~\ref{alg:smrapi}.
\hyaline{} keeps track of all active threads using \emph{special}
reference counters. Those counters are not typical and do not represent
the number of references to objects directly.
Unlike traditional per-object counting, the use of
reference counters is triggered only when handling
retired objects (nodes). Thus, insertions and read-only traversals
avoid expensive (and inconvenient) per-access counting.

\newcommand{\smrapi}{%
{
\LinesNumberedHidden
\begin{algorithm2e}[H]
 handle\_t Handle = \textbf{enter}()\;
 \tcp{\textbf{deref} is for \hyaline{}-S,}
 \tcp{not needed in \hyaline{}}
 \HiLi List = \textbf{deref}(\&LinkedList)\;
 \HiLi Node = \textbf{deref}(\&List->Next)\;
 \textbf{retire}(Node)\;
 \tcp{Do something else...}
 \textbf{leave}(Handle)\;
 \tcp{\textbf{Transparency:} the thread}
 \tcp{need not check any of the}
 \tcp{retired nodes after this}
\end{algorithm2e}
}
}

\begin{figure}
\smrapi
\vspace{-15pt}
\caption{\hyaline{}'s transparent API.}
\label{alg:smrapi}
\end{figure}%

\subsection{Simplified Version}
\label{sec:scalar}

We first describe a simpler version of \hyaline{} that manipulates only a single \emph{retirement list}. This version is more prone to CAS contention, a problem addressed by a scalable version that we present in Section~\ref{sec:vectorized-smr}. 

\hyaline{}'s key idea is that all threads participate in the tracking of  retired
nodes in the global list even if they are not actively retiring any nodes themselves.
A special {\tt Head} tuple is associated with the retirement list.
The tuple consists of {\tt HPtr} and {\tt HRef} fields. {\tt HPtr} is a pointer
to the beginning of the list, and {\tt HRef} counts the number of active threads.
Initially, when the list is empty, {\tt HPtr = Null} and {\tt HRef = 0}.

When each thread \emph{enters},
it atomically increments the {\tt HRef} field to indicate that a new
thread has arrived. At the same time, the thread records
a snapshot value of \texttt{HPtr} at the moment it entered.
The thread stores this snapshot value
in a special per-thread {\tt Handle} variable.
Since updates on the {\tt [HRef,HPtr]} tuple have to be atomic, we use double-width RMW
to update \texttt{Head}.

As nodes are \emph{retired}, threads append them to the list
(Figure~\ref{fig:smrintro}).
Since nodes need to be connected in the list,
each node incorporates a special header in
addition to any other fields used for representing the encapsulating
data structure. The header contains {\tt Next} and {\tt NRef} fields.
{\tt Next} is a pointer
to the next node in the list, and {\tt NRef} of every non-{\tt Head} node
counts threads that can still access this node.
For the very first node, {\tt HRef} itself serves this purpose. (We will describe how
{\tt NRef} is initialized later.)

\begin{figure}
\centering
\includegraphics[width=.8\columnwidth]{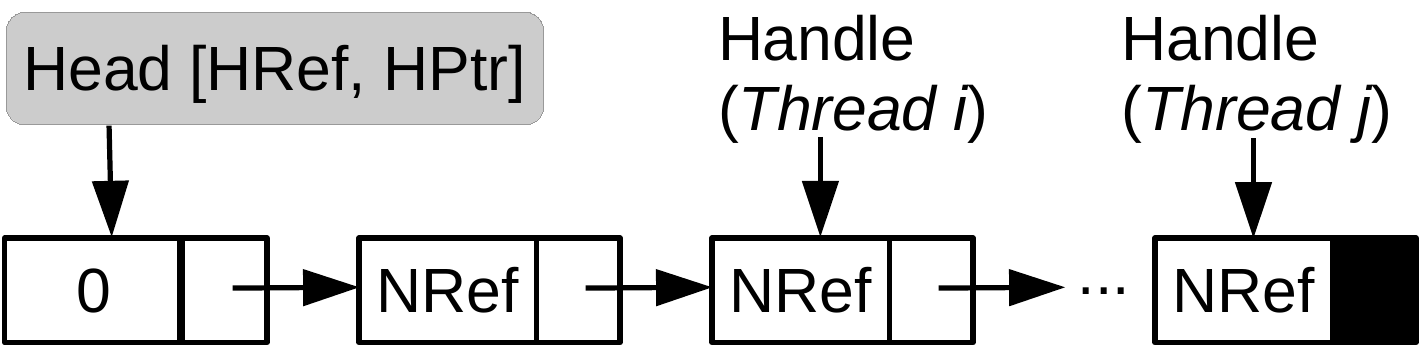}
\caption{(Simplified) \hyaline{}: List of retired nodes.}
\label{fig:smrintro}
\end{figure}%

When a thread completes a data structure operation (\emph{leav\-es}),
it decrements {\tt HRef} to indicate that one
thread has just left. Simultaneously, it retrieves the {\tt HPtr} pointer
and then traverses a sublist of nodes from {\tt HPtr} to {\tt Handle} that were retired
since it initiated the operation (\textit{enter}).
While traversing the sublist, the thread decrements {\tt NRef} counters
for every non-Head node. The first node's counter ({\tt HRef}) is already
decremented. A node is freed when its counter becomes $0$.

Using \texttt{HRef} for the first node
prevents the ABA problem as a thread has a reference to every node through its {\tt Handle} inclusively -- i.e., no other thread can recycle these nodes.

We now describe how {\tt NRef} values get propagated across the list.
When retiring a node, its {\tt NRef} is set to $0$, as the actual
counter for the very first node in the list is
inferred from the {\tt HRef} field of {\tt Head}.
As threads insert retired nodes,
they initialize {\tt Next} and atomically update {\tt Head}
(shown with dashed lines in Figure~\ref{fig:smrfig}, part (a)).
{\tt NRef} of the predecessor is initially 0. However, since
a node was just added, the predecessor is no longer the first
node, and any concurrent thread may decrement its {\tt NRef}, 
converting it to a negative value, but will not deallocate this predecessor
({\tt NRef} must become 0 for the node to be deallocated).
Finally, the current thread atomically adds the snapshot value of {\tt HRef}
(obtained while appending the new node) to the {\tt NRef} field
of the predecessor, and its new \emph{adjusted} value becomes~$\ge 1$ (Figure~\ref{fig:smrfig}, part (b)). Retiring is now complete (Figure~\ref{fig:smrfig}, part (c)).

\begin{figure}
\centering
\includegraphics[width=\columnwidth]{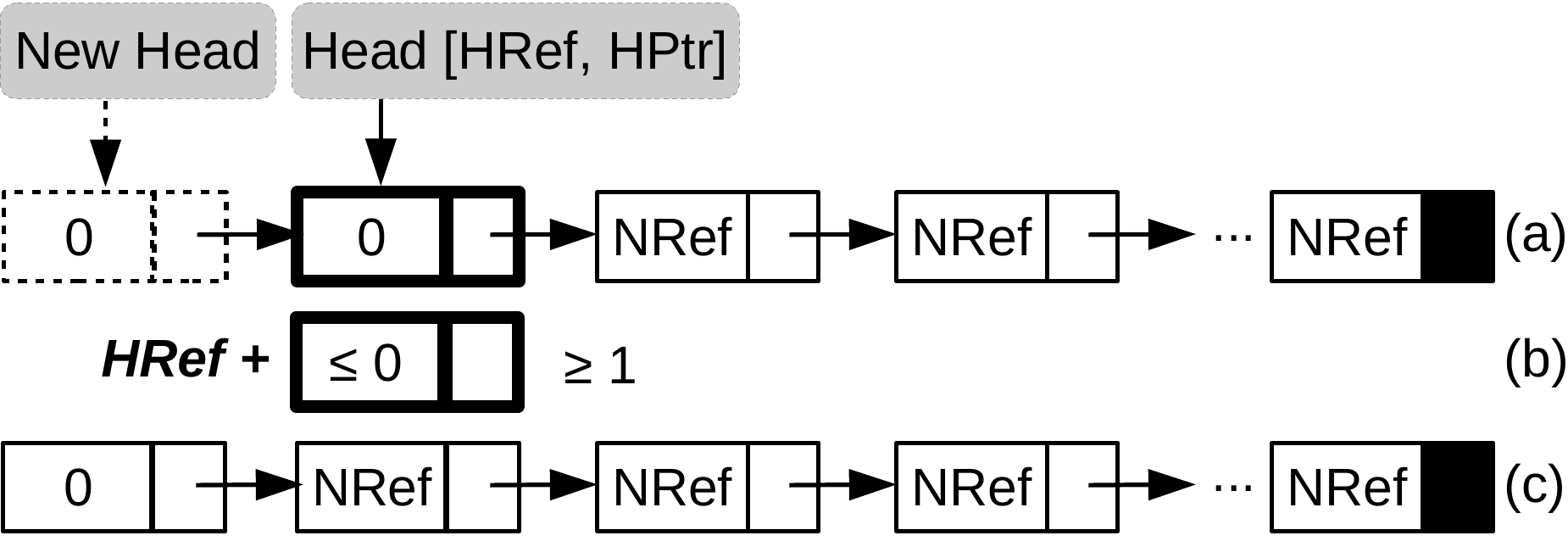}
\caption{(Simplified) \hyaline{}: Adjusting the reference counter of a predecessor node when retiring.}
\label{fig:smrfig}
\end{figure}

Essentially, {\tt NRef} is the difference between two logical variables:
the number of times the node is acquired and the number of times it is released. Due to concurrency interleaving, {\tt NRef} is relaxed and can be negative.

Figure~\ref{fig:smrexample} shows an example with 3 threads.
Initially, \texttt{HRef} is $0$ and \texttt{HPtr} is $Null$. (a) Thread 1 calls
\textit{enter} to atomically increment {\tt HRef} and retrieve its
handle. (b) Thread 1 retires node N1; as the list is empty,
there is no predecessor to adjust. (c) Thread 2 enters, but (d) it stalls
while retiring N2.
Meanwhile, (e) Thread 3 enters.
(f) Thread 1 leaves and dereferences
all nodes in the list through its handle {\tt Null}.
Since N2 is the first node, {\tt HRef} is decremented, but N2's
{\tt NRef} field remains intact. N1 stays as its {\tt NRef} is now negative.
(g) Thread 2
resumes and completes its adjustment for N1. (h) Thread 2 then leaves,
dereferences all nodes, and deallocates N1. (i) Finally, Thread 3
leaves and deallocates N2.

\begin{figure}
\includegraphics[width=.9\columnwidth]{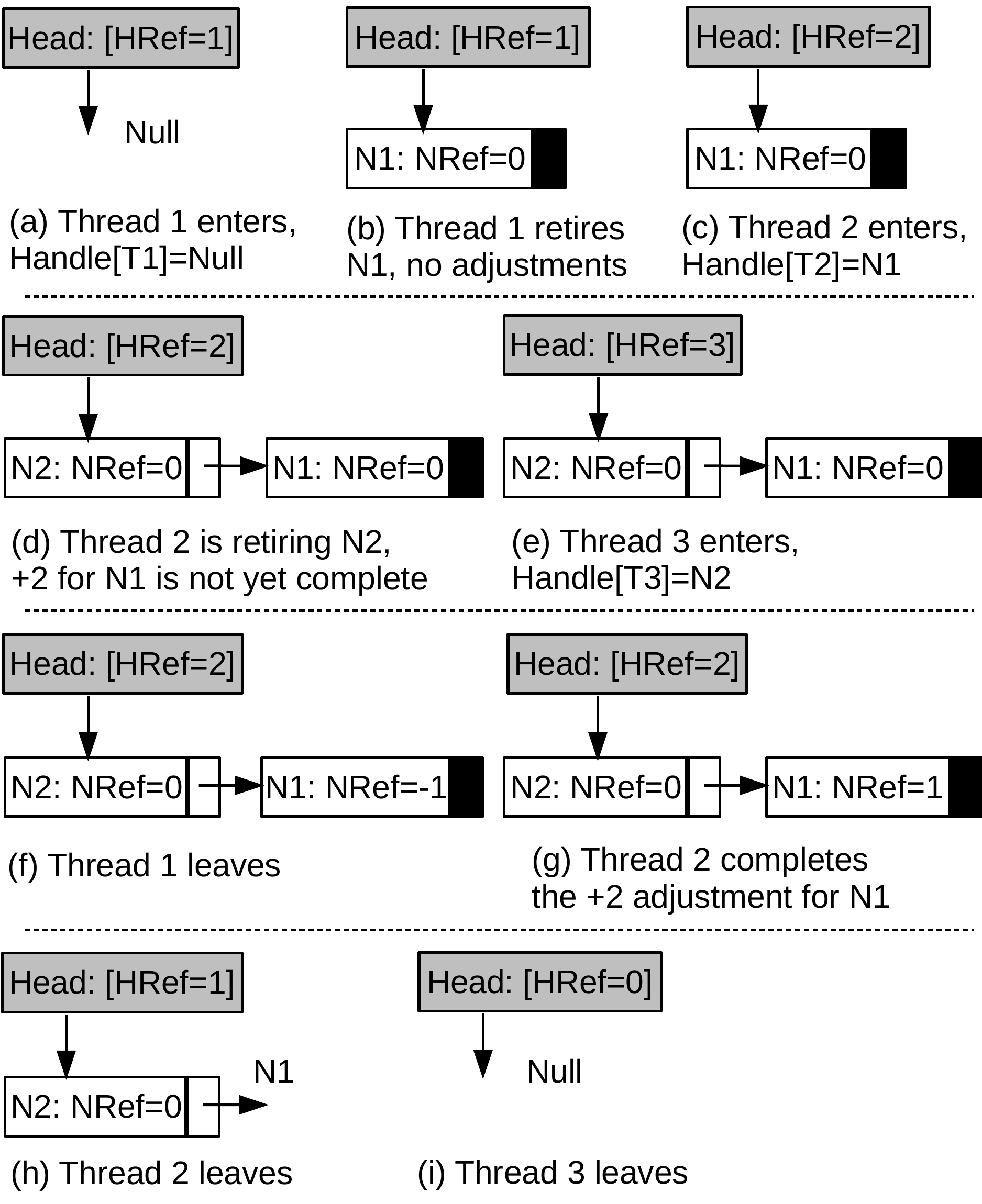}
\caption{Example for single-list \hyaline{}, Nx is a node.}
\label{fig:smrexample}
\end{figure}

Although this version is not yet optimized for performance, we make one
important observation regarding the algorithm's asynchronous tracking mentioned in introduction:
threads traverse lists just once when dereferencing nodes in \textit{leave}.
This is unlike EBR, where all threads are periodically checked
if they are past the retired node(s) epoch.
Section~\ref{sec:eval} reveals this to be \hyaline{}'s advantage for oversubscribed tests.

\subsubsection*{Alternative Designs}
\hyaline{} carefully avoids the ABA problem which is possible with other
designs. A straightforward alternative is to use {\tt NRef} in the first node
as usual, i.e., to indicate its reference counter. {\tt HRef} can
then indicate the reference counter of the node to be retired in
the future. However, this design triggers the ABA problem, as the node pointed
to by {\tt Handle} (end-of-list marker) can be recycled and may reappear in the
retirement sublist when \textit{leave} is called. The sublist
will get truncated, and the remaining nodes will never be dereferenced.
An extra ABA tag could prevent this problem but {\tt Head} already
needs double-width operations.

{\tt NRef} could also store the reference counter
of the next rather than the current node. {\tt HRef} would indicate the
reference counter of the first node as in \hyaline{}.
While this design is ABA-free, it has a major drawback: nodes must
be dereferenced in the reversed order, as they are dependent on
each other. It complicates the implementation, as backward links also
need to be stored. Moreover, {\tt Head} cannot be updated until the retirement
sublist is fully traversed, and by that time, other
nodes may already be retired. If deallocation is slow, one
unlucky thread can easily get stuck in a state where it has to constantly
deallocate nodes retired by other threads. Other threads will simply
continue dereferencing their counters and keep retiring more and more nodes.
\hyaline{} avoids this problem by immediately
decrementing {\tt HRef}, as all nodes are independent from each other there.

\subsection{Scalable Multiple-List Version}
\label{sec:vectorized-smr}

If deletions are frequent, \textit{retire} calls may create contention on
{\tt Head}. To alleviate the contention, threads create local lists
of nodes to be retired. Threads retire entire batches of nodes and keep a single
reference counter per batch rather than individual nodes. Batches do not have direct analogues in epoch-based approaches, where
all retired nodes are always in thread-local lists. However, batch size (the number of nodes in a batch) impacts
the cost of retirement in a way that is similar to the frequency of epoch counter increments.

Frequent \textit{enter} and \textit{leave} calls also create contention on
{\tt Head}, which is undesirable for high-throughput data structures. To address
this problem, we introduce the concept of \emph{slots}, which a thread
chooses randomly or based on its ID. Slots do not need to be statically
assigned: they can change from one operation to another.
Each slot has its own {\tt Head}, and thus, we end up with multiple
retirement lists. When a batch is retired, it needs to be added
to each slot that has its {\tt HRef}~$\ne 0$ (i.e., slots with
active threads).

Since batches
are added atomically only to one slot at a time, slots may end up with
non-identical order of batches. To support this, we need individual
list pointers. We take advantage of the fact that we retire entire batches
rather than nodes. To that end, we require the number of
nodes in batches to be strictly greater than the number of slots. Each node
in a batch keeps the {\tt Next} pointer for the corresponding slot's retirement
list, and one additional node will keep the per-batch
{\tt NRef} counter instead. Additionally, all nodes in the batch are linked
together, and each node has an extra pointer to the node with {\tt NRef}.
Thus, each node keeps three variables irrespective of
batch sizes and total number of slots.

In Section~\ref{sec:scalar}, we described reference adjustments using
signed integer arithmetic. The same argument applies to unsigned
numbers, in which case negative numbers represent very large
integers. We generalize this idea to accommodate \hyaline{}'s multiple-list
version. When adjusting a predecessor in
slot $i$, we add $\mathit{Adjs}+\mathit{HRef}_i$ rather than just $\mathit{HRef}_i$,
where $\mathit{Adjs}$ is a special constant which prevents
the adjustment for the predecessor to complete until all slots are handled.
Assuming that the number of slots, $k$, is a power of 2, and
the maximum representable unsigned integer value is $2^N-1$, we calculate:
$\mathit{Adjs} = \floor*{\frac{2^N-1}{k}} + 1$.

For example, if $k=1$ (simple version), $Adjs$ cancels out right away.
When $k=8$, assuming 64-bit integers, $\mathit{Adjs} = 2^{61}$.
It is easy to see that more generally, $\mathit{Adjs}$ cancels out after $k$ additions: $k\times\mathit{Adjs}=0$ due to unsigned integer overflow.

When retiring a batch, a predecessor batch has to accumulate $\mathit{Adjs}$
for all $k$ lists for the adjustment to complete.
Since some slots have no active threads, we accumulate $\mathit{Adjs}$ for
them when inserting the batch, and atomically
add the net value to {\tt NRef} of the \textit{current} batch as the final step.

\begin{figure}
\includegraphics[width=.9\columnwidth]{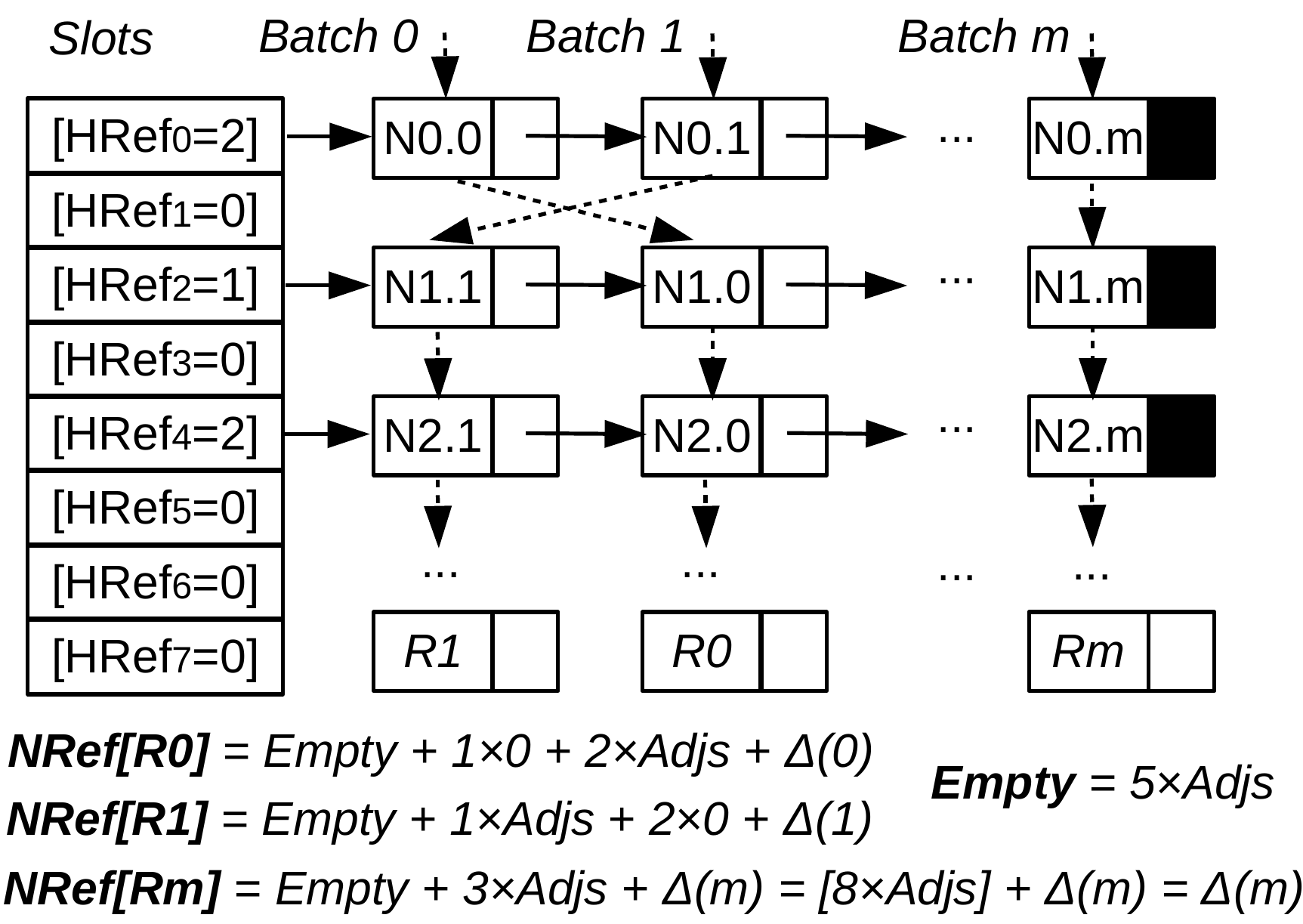}
\caption{Scalable (multiple-list) \hyaline{}.}
\label{fig:smrvector}
\end{figure}

In Figure~\ref{fig:smrvector}, we present an example with $k=8$ slots.
For the purpose of this example, we enumerated nodes to reflect their relative
slot positions (skipping empty slots) and corresponding batch
numbers. For convenience, $\mathit{Empty}$ denotes adjustments for five empty
slots ($\mathit{5\times Adjs}$) that need not be handled.
Batches are added one slot at a time, and two concurrent threads insert
them in an interleaved fashion.
When \emph{Batch~0} is inserted, it ends up in the first position
for slot~0 and the second position in all other active slots. {\tt NRef}
for Batch~0 is stored in the node $\mathit{R0}$ and contains
$\mathit{Empty}$ for empty slots, 0 for slot~0 (not yet adjusted), $2\times\mathit{Adjs}$ for slots~2 and~4
(adjusted when retiring \emph{Batch~1}), and the
\emph{actual} counter component $\mathit{\Delta}_0$. $\mathit{\Delta}_0$
contains the snapshot values of $\mathit{HRef}_2$ and $\mathit{HRef}_4$
when Batch~1 is inserted.
A similar breakdown is shown for \emph{Batch~1}.
For $\mathit{Batch}~m$, all adjustments are already
cancelled out, and its NRef node contains just $\mathit{\Delta}_m$.

\subsubsection*{\hyaline{}-1 for Single-width CAS}
In a special case, if every thread allocates its own unique slot, we can
squeeze {\tt HRef} into one bit and merge it with
{\tt HPtr}. This simplifies \textit{enter} and \textit{leave}, since they
can use ordinary \textit{write} and \textit{swap} in lieu of CAS,
making these operations wait-free.
Adjustments
can be also simplified: instead of adjusting predecessors and empty slots,
we count the number of slots a batch is added to.  ($k$ does not have to
be a power of 2.)
After adding the batch to the last slot, {\tt NRef} of the batch
is adjusted by this counter.

\subsubsection*{Contention}
Assuming that slots are cache-line aligned, CAS on {\tt Head} is almost uncontended, and
MESIF/MOESI protocols used by modern Intel/AMD CPUs incur no substantial performance
penalty~\cite{Schweizer:2015:ECA:2923305.2923811} (contrary to popular belief). The cost of {\it enter} and {\it leave} is therefore relatively small.
Section~\ref{sec:eval} further shows that there is very
little difference in \hyaline{}'s or \hyaline{}-1's overall performance even for
high-throughput data structures, which confirms that CAS in \textit{enter} or \textit{leave} is not a source of any measurable performance penalty.

\subsubsection*{Costs}
Hyaline-1's and EBR's \textit{enter}/\textit{leave} costs are similar at the very least for x86 (write+barrier is replaced with \textit{swap}, i.e., \textit{xchg}, by recent gcc/clang compilers; AMD explicitly recommends \textit{xchg} for sequentially-consistent writes~\cite{AMD}). Due to low CAS contention (used in lieu of \textit{swap}), more general Hyaline also exhibits very similar performance. Though \textit{leave} is longer in Hyaline and Hyaline-1 due to list traversals, this cost is simply incurred elsewhere in EBR, e.g., in \textit{retire}.

\section{Algorithm Descriptions}
\label{sec:algorithms}

The main idea of Hyaline is that it always implicitly keeps track of concurrent threads. The number of concurrent threads gets reflected in the counter for each retired batch. Each of the concurrent threads has to decrement this counter explicitly. When the counter gets to 0, the batch is reclaimed.

Figure~\ref{fig:hyalineds} presents the layout of nodes, batch structure, and global state. Batches are first accumulated locally.
All retired nodes are linked together using \texttt{BatchNext}.
The very last node (with a reference counter) is denoted as \texttt{NRefNode}, its \texttt{BatchNext} points to the very first node (i.e., in a cyclic list manner).
Every node has a pointer to \texttt{NRefNode}.

\begin{figure}[ht]
\begin{subfigure}{.47\columnwidth}
\lstset{language=C++}
\begin{lstlisting}[frame=single]
// Double-width integer
struct head_t {
	int     HRef;
	node_t *HPtr;
};
// Nodes are retired locally
// by appending to a batch;
// when the batch fills up,
// retire() is called for the
// entire batch. As nodes get
// appended to the batch,
// MinBirth is set to minimum
// BirthEra across batch nodes.
struct local_batch_t {
	// Pointer to the NRef node
	node_t *NRefNode;
	// The first node in the batch
	node_t *FirstNode;
	// Only for Hyaline-S or -1S
	int     MinBirth;
};
\end{lstlisting}
\end{subfigure}
\hspace{.5mm}
\begin{subfigure}{.47\columnwidth}
\lstset{language=C++}
\begin{lstlisting}[frame=single]
struct node_t {
	union {
		// NRefNode: refer. counter
		int     NRef;
		// Others: per-slot list
		node_t *Next;
		// Only for Hyaline-S or -1S
		int     BirthEra;
	};
	// Pointer to the NRef node
	node_t *NRefNode;
	// Next node in the batch
	node_t *BatchNext;
};

// Slots (all Hyaline variants)
head_t Heads[k];
// Last era (Hyaline-S or -1S)
int Accesses[k];
// Acknowledgments for Hyaline-S
int Acks[k];
\end{lstlisting}
\end{subfigure}
\caption{Hyaline's nodes and global state.}
\label{fig:hyalineds}
\end{figure}

\subsection{Basic \hyaline{}}

In Figure~\ref{alg:hyaline}, we present pseudocode for the \textit{enter},
\textit{leave}, and \textit{retire} operations. \textit{enter}
atomically increments the \texttt{HRef} variable while fetching the current pointer in
a given slot.
\textit{retire} inserts a batch to all slots. For empty slots, it counts
$\mathit{Empty}$ adjustments and adds $\mathit{Empty}$ to {\tt NRef} of the batch in the very end.
For each slot, a predecessor is adjusted by the corresponding {\tt HRef} snapshot
value and $\mathit{Adjs}$.
\textit{leave} decrements \texttt{HRef}, but also reads \texttt{Next} from the
node \texttt{Head} is pointing to. Since a thread always has a reference
to the head of the list,
reading the first node is safe. The last thread replaces the first node
with \texttt{Null} treating it as a predecessor
in \textit{retire}. Finally, succeeding nodes (if any) are dereferenced
in the \textit{traverse} helper method.

\hyaline{}-1 in Figure~\ref{alg:hyaline1} replaces {\it enter} and
{\it leave} with simpler equivalents. Since one thread is the sole
owner of all nodes, \textit{leave} can detach the first node
immediately and read the node that follows after that. This simpler scheme
also does not adjust predecessor nodes.

\newcommand{\algoone}{%
\begin{algorithm2e}[H]
\Fn{\textbf{handle\_t} enter(\textbf{int} slot)} {
Heads[slot] = \{ .HRef=1, .HPtr=Null \}\;
\Return Null\tcp*{Returns a handle}
}
\tcp{\textbf{Replace \#2\# in retire() with:} Inserts++}
\end{algorithm2e}
}

\newcommand{\algotwo}{%
\begin{algorithm2e}[H]
\setcounter{AlgoLine}{3}
\Fn{\textbf{void} leave(\textbf{int} slot, \textbf{handle\_t} handle)} {
Head = SWAP(\&Heads[slot], \{ .HRef=0, .HPtr=Null \})\;
\lIf {Head.HPtr\ $\neq$ Null} {traverse(Head.HPtr, handle)}
}
\tcp{\textbf{Replace \#3\#:} adjust(batch->FirstNode, Inserts)}
\end{algorithm2e}
}

\newcommand{\algsone}{%
\begin{algorithm2e}[H]
\textbf{int} AllocEra = 0\tcp*{Initialization}
thread \textbf{int} AllocCounter = 0\;
\lForAll{\textbf{int} Access $\in$ Accesses[k]} {Access = 0}
\lForAll{\textbf{signed int} Ack $\in$ Acks[k]} {Ack = 0}
\BlankLine
\Fn{\textbf{node\_t} *deref(\textbf{int} slot, \textbf{node\_t} **ptr\_node)} {
Access = Accesses[slot]\;
\While {True} {
\textbf{node\_t} * Node = (*ptr\_node)\;
Alloc = AllocEra\;
\lIf {Access = Alloc} {\Return Node}
Access = touch(slot, Alloc)\;
}
}
\Fn{\textbf{void} retire(\textbf{local\_batch\_t} *batch)} {
\tcp{\textbf{Replace \#1\# in retire() with:}}
Access = Accesses[slot]\;
\lIf {Head.HRef = 0 \OrOp Access < batch->MinBirth} {...}
\tcp{\textbf{Place after \#2\# in retire():}}
FAA(\&Acks[slot], Head.HRef)\tcp*{\textbf{\hyaline{}-S only}}
}
\end{algorithm2e}
}

\newcommand{\algstwo}{%
\begin{algorithm2e}[H]
\setcounter{AlgoLine}{15}
\Fn{\textbf{void} init\_node(\textbf{node\_t} *node)} {
\lIf {AllocCounter++ \ModOp Freq = 0} {FAA(\&AllocEra, 1)}
node->BirthEra = AllocEra\tcp*{Shared with Next}
}
\Fn{\textbf{int} touch(\textbf{int} slot, \textbf{int} era)} {
\Do (\tcp*[f]{\textbf{\hyaline{}-1S: use a regular memory write}}){\Not CAS(\&Accesses[slot], Access, era)}{
Access = Accesses[slot]\;
\lIf {Access $\ge$ era} {
\Return Access}
}
\Return era
}
\tcp{\textbf{Hyaline-S only, not needed in Hyaline-1S}}
\Fn{\textbf{handle\_t} enter(\textbf{int} *slot)} {
\While {Acks[*slot] $\ge$ Threshold} {
*slot = (*slot + 1) \ModOp k\tcp*{Try all k slots}
}
}
\Fn{\textbf{void} traverse(\textbf{int} slot, \textbf{node\_t} *next, \textbf{handle\_t} handle)} {
Counter = 0\;
\lDo{...} {
Counter++ ...
}
FAA(\&Acks[slot], $-$Counter)\;
}
\end{algorithm2e}
}

\newcommand{\algone}{%
\begin{algorithm2e}[H]
\ForAll (\tcp*[f]{Initialization}) {\textbf{head\_t} Head $\in$ Heads[k]} {
Head.HRef = 0, Head.HPtr = Null\;
}
\Fn{\textbf{handle\_t} enter(\textbf{int} slot)} {
Last = FAA(\&Heads[slot], \{ .HRef=1, .HPtr=0 \})\;
\Return Last.HPtr\tcp*{Returns a handle}
}
\Fn{\textbf{void} leave(\textbf{int} slot, \textbf{handle\_t} handle)} {
\Do(\tcp*[f]{Decrement HRef and fetch Next}){\Not CAS(\&Heads[slot], Head, New)}{
Head = Heads[slot]\;
Curr = Head.HPtr\;
\If {Curr $\neq$ handle} {
	Next = Curr->Next\;
}
New.HPtr = Curr\;
\lIf {Head.HRef = 1} {New.HPtr = Null}
New.HRef = Head.HRef - 1\;
}
\If (\tcp*[f]{Treat Curr as if}) {Head.HRef = 1 \AndOp Curr} {
    adjust(Curr, Adjs)\tcp*{it were a predecessor}
}
\If (\tcp*[f]{Non-empty list}) {Curr $\neq$ handle} {
	traverse(Next, handle)\;
}
}
\Fn{\textbf{void} adjust(\textbf{node\_t} *node, \textbf{int} val)} {
Ref = node->NRefNode\;
\tcp{free\_batch() frees all nodes by iterating}
\tcp{BatchNext. BatchNext of Ref points to the}
\tcp{first node in the batch.}
	\lIf {FAA(\&Ref->NRef, val) = -val} {free\_batch(Ref->BatchNext)}
}
\end{algorithm2e}
}

\newcommand{\algtwo}{%
\begin{algorithm2e}[H]
\setcounter{AlgoLine}{22}
\Fn{\textbf{void} retire(\textbf{local\_batch\_t} *batch)} {
doAdj = False, Empty = 0, Inserts = 0\;
CurrNode = batch->FirstNode\;
batch->NRefNode->NRef = 0\;
\ForAll{\textbf{int} slot $\in$ 0..k-1} {
\Do (\tcp*[f]{Add the batch to this slot}) {\Not CAS(\&Heads[slot], Head, New)}{
Head = Heads[slot]\;
\If (\tcp*[f]{\textbf{\#1\#}}) {Head.HRef = 0} {
doAdj = True, Empty += Adjs\;
\textbf{continue} with the next \textit{slot};
}
New.HPtr = CurrNode\;
New.HRef = Head.HRef\;
New.HPtr->Next = Head.HPtr\;
}
CurrNode = CurrNode->BatchNext\;
adjust(Head.HPtr, Adjs + Head.HRef)\tcp*{\textbf{\#2\#}}
}
\lIf (\tcp*[f]{\textbf{\#3\#}}) {doAdj} {adjust(batch->FirstNode, Empty)}}
\Fn{\textbf{void} traverse(\textbf{node\_t} *next, \textbf{handle\_t} handle)} {
\Do (\tcp*[f]{Traverse the retirement sublist}) {Curr $\neq$ handle}{
Curr = next\;
\lIf {Curr = Null} {\textbf{break}}
next = Curr->Next\;
Ref = Curr->NRefNode\;
\If (\tcp*[f]{BatchNext of Ref points to the}) {FAA(\&Ref->NRef, -1) = 1} {
		free\_batch(Ref->BatchNext)\tcp*{first node in the batch}
}
}
}
\end{algorithm2e}
}

\begin{figure*}
\begin{subfigure}{.5\textwidth}
\algone
\end{subfigure}%
\begin{subfigure}{.5\textwidth}
\algtwo
\end{subfigure}%
\caption{\hyaline{} for double-width CAS.}
\label{alg:hyaline}
\end{figure*}

\begin{figure*}
\begin{subfigure}{.5\textwidth}
\algoone
\end{subfigure}%
\begin{subfigure}{.5\textwidth}
\algotwo
\end{subfigure}%
\caption{\hyaline{}-1 for single-width CAS (substitutes for functions).}
\label{alg:hyaline1}
\end{figure*}

\begin{figure*}
\begin{subfigure}{.5\textwidth}
\algsone
\end{subfigure}%
\begin{subfigure}{.5\textwidth}
\algstwo
\end{subfigure}%
\caption{\hyaline{}-S and \hyaline{}-1S: dealing with stalled threads (extension).}
\label{alg:hyalines}
\end{figure*}

\subsection{\hyaline{}-S}

\label{sec:hyalinerobust}

To deal with stalled threads in \hyaline{}, we extend \hyaline{} by partially adopting the idea from HE and IBR to
record \emph{birth eras} when allocating memory.
The high-level idea is to mark each allocated object with the global counter, so that stalled threads will only hold older objects. Objects allocated
after those threads stall will be counted only towards active threads.
Birth eras simply facilitate
detection of stalled threads in Hyaline. (Compare it to HE and IBR,
where birth and retire eras define actual
reclamation intervals.)
Unlike HE/IBR, birth eras share space with other variables, e.g., \texttt{Next}, as
they are not required to survive \textit{retire}.

\hyaline{}-S, unlike \hyaline{}-1S,
supports multiple threads per each slot, so we
have to record eras such that they can be shared across multiple
threads. That presents extra challenges when dealing with stalled
threads since they may interleave with non-stalled threads.

In Figure~\ref{alg:hyalines}, we present \hyaline{}-S.
Our API model is reminiscent of 2GE-IBR~\cite{IntervalBased} which only requires
to additionally wrap all pointer reads in a special \textit{deref} call.
The eras are 64-bit numbers
which are assumed to never overflow in practice.
When nodes are allocated, \textit{init\_node} initializes their birth eras
with the era clock value. When dereferencing pointers, threads call
\textit{deref} to update a per-slot \textit{access} era value.
Since \hyaline{}-S allows arbitrary number of threads per slot,
threads must share per-slot eras, and the maximum era
needs to be set using the \textit{touch} helper function.
(\hyaline{}-1S can just write the new era, as there is a 1:1
thread-to-slot mapping.)
Since all active threads update
eras when calling \textit{deref} in their slots, \textit{retire} simply uses
the minimum birth era across all nodes in a batch, and skips slots with
stale eras.

Since threads share per-slot eras in \hyaline{}-S, it is crucial to stay away from slots
occupied by stalled threads when entering.
Each slot keeps a special \texttt{Ack} value incremented
by \textit{retire}. \texttt{Ack} accumulates the  total number of active threads across
all retired batches in the slot. Since all retired batches inevitably appear in the
retirement sublists of all active threads, each thread acknowledges that it no longer
references batches by decrementing \texttt{Ack} in \textit{traverse}.
\texttt{Ack} can be negative temporarily if \textit{traverse}
takes place before \textit{FAA} in \textit{retire}. (Nonetheless, \texttt{Ack} only increases after finite number of retirements when at least one
thread is stalled, i.e., it does not call \textit{traverse}.) \texttt{Ack} may also
be positive, but after some threshold (e.g., 8192), \textit{enter} can assume that
the corresponding slot is occupied by stalled threads.
\texttt{Ack}s do not incur any measurable penalty as evidenced
by Section~\ref{sec:eval} where \hyaline{}-S and \hyaline{}-1S have roughly similar performance.

All slots may end up being occupied by stalled threads in Hyaline-S. To guarantee robustness, we can adaptively increase the number of slots
by using an extra array which stores pointers to arrays of slots (Section~\ref{sec:robust}).

\subsection{Adaptive Resizing for Hyaline-S}
\label{sec:robust}

Unlike Hyaline-1S which allocates a dedicated slot for each thread and is fully robust, Hyaline-S caps
the total number of slots. This limits robustness guarantees for Hyaline-S in rare situations
when all slots fill up with stalled threads and they begin to interfere with active threads.

We now describe an approach which makes Hyaline-S fully robust by adaptively increasing the
number of available slots, $k$, as a larger number of threads are stalled. We denote the initial $k$ value (a constant), $\mathit{Kmin}$. The current $k$ value is stored in a global atomic variable.

When a batch is finalized and \textit{retired}, we read the current $k$ value. (There is no problem if
concurrent threads increase the $k$ value right after we read it, as new slots will be used
by new \textit{enter} calls which need not account for already retired nodes. A larger
than necessary $k$ is also not a problem since the batch will simply be added to extra slots.)
We calculate the $\mathit{Adjs}$ value based on the current $k$ value and store it in each batch. Each node in a batch
contains a pointer to \texttt{NRefNode}, but \texttt{NRefNode} itself does not need to keep
this pointer. Instead, we use this variable to store the current $\mathit{Adjs}$ value for the batch.

When calling \textit{adjust}, we use the corresponding batch's $\mathit{Adjs}$ value. In Figure~\ref{alg:hyaline}, we have three \textit{adjust} calls: Line~17 uses $\mathit{Adjs}$ for the $\mathit{Curr}$'s batch,
Line~38 uses $\mathit{Adjs}$ for \texttt{HPtr}'s batch, and Line~39 uses $\mathit{Adjs}$ for the current batch.

When stalled threads occupy all slots (Figure~\ref{alg:hyalines}, Line~26),
we adaptively increase the number of slots. Since we cannot resize the initial array of slots easily, we
maintain a \textit{directory of slots}, an array of pointers to arrays of slots, as shown in Figure~\ref{fig:robustdir}. This array is fixed-size and small, e.g., for 64-bit CPUs, it
never exceeds 64 entries. Initially, only index $0$ points to the array of slots with $\mathit{Kmin}$ entries. As \textit{enter} runs out of slots, we allocate an additional array of $(\mathit{2\times Kmin}-\mathit{Kmin})$ slots such that the total number of slots doubles. We atomically
change index $1$ to point to the new array (we also offset this pointer by $\mathit{Kmin}$ to simplify the slot position calculation). If a concurrent thread also changes index $1$, the thread
for which the corresponding CAS fails will discard the allocated buffer. The aforementioned procedure applies to all arrays which use slots: \textit{Heads}, \textit{Accesses}, and \textit{Acks}.

\begin{figure*}
\includegraphics[width=.7\textwidth]{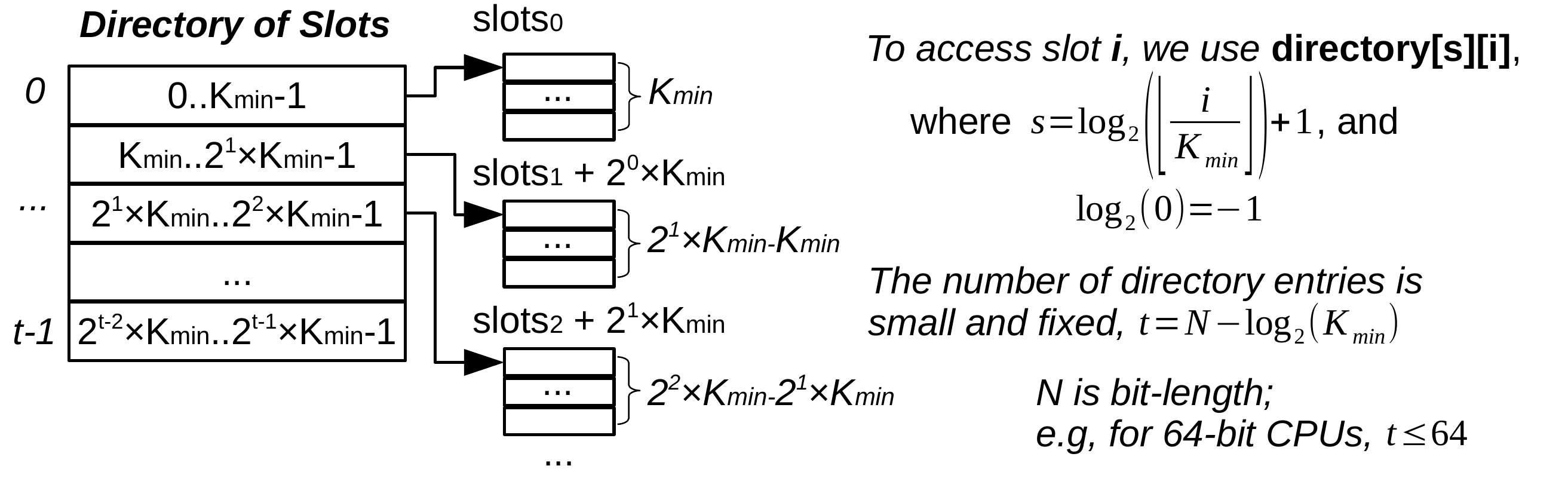}
\caption{Hyaline-S: adaptive resizing.}
\label{fig:robustdir}
\end{figure*}

To access a slot, we use the formula from Figure~\ref{fig:robustdir}, which calculates a directory
array index. The $\mathit{log}_2$ operation, including a special case of $\mathit{log}_2(0)=-1$, is efficiently implemented by the \textit{leading zero count} instruction, available on modern CPUs,
by using $\mathit{log}_2(x) = N - \mathit{lzcnt}(x) - 1$, where $N$ is bit-length.
Since we always double the number of slots, $k$, and the initial $\mathit{Kmin}$ value is a power-of-two number, our assumption that $k$ is a power-of-two number is still valid.

We always increase the number of slots as we detect more stalled threads and run out of slots. However, the number
of slots is bounded by the total number of stalled threads (rounded to the next power-of-two number). Since the number of threads is finite, memory occupied by slots is bounded, i.e., our algorithm is still robust. Existing robust SMR schemes similarly
require dedicated slots per \textit{each} thread.

\subsection{Usage Preference}
It should be feasible to always use Hyaline-1 in lieu of EBR, and Hyaline-1S
in lieu of HE, HP, or IBR given Hyaline's performance benefits (Section~\ref{sec:eval}). One exception is when users deliberately want to avoid reclamation by read-only threads due to some extremely rigid latency requirements. This scenario seems uncommon for general-purpose systems, and it would not improve the overall throughput anyhow.

Hyaline-1 and Hyaline-1S
are very portable and expose a relatively simple API. Hyaline and Hyaline-S provide full transparency but additionally require LL/SC or double-width CAS, which degrades portability. All Hyaline schemes simplify integration, e.g., it is much easier to register/unregister threads dynamically than with the aforementioned schemes. Garbage collectors have different trade-offs, and Hyaline's applicability in the corresponding applications is similar to that of EBR, HE, HP, and IBR.

\section{Correctness}
\label{sec:correctness}

We now prove correctness, lock-freedom, and robustness.

\begin{theorem}
All \hyaline{} variants are reclamation-safe.
\end{theorem}

\begin{proof}
In a correct program, a retired batch cannot be accessed by subsequent operations. Only concurrent operations may still access it. Each of those concurrent operations starts by calling \hyaline{}'s \textit{enter}. Any batch retired during this concurrent execution will have its $\mathit{NRef}\not=1$ (Lines~26~and~39). If another thread executes \textit{leave} after Line 36 and before the last adjustment in Line 39, then it will start by decrementing the retired object's reference count such that it will be a very large number (Line~46). Only objects with a new reference count of zero are reclaimed (Lines~22~or~46). Thus, those retired objects with very large reference counts are safe from being reclaimed. After executing Line~39 across all slots where the batch is placed, the retired object's reference count will reflect the correct number of concurrent threads that have not executed \textit{leave} yet. Hence, the object will not be reclaimed until all these threads execute \textit{leave}.

\hyaline{}-(1)S, regardless of $\mathit{HRef}$ values, skips slots with eras
that are smaller than $\mathit{min\_birth}$ from a retired batch. $\mathit{min\_birth}$
signifies the oldest node in the batch.
\textit{deref} always updates per-slot eras to keep them in sync
with the global era clock. Thus, the retired batch must have been covered
by per-slot eras unless none of its nodes is ever dereferenced.
\end{proof}

\begin{theorem}
All \hyaline{} variants are lock-free. (With respect to CPU progress
only, see Theorem~\ref{theor:robust} for robustness.)
\end{theorem}

\begin{proof}
\hyaline{} has two unbounded loops (Lines~7-15 and 28-36). If the CAS operation fails in the first loop (Lines~7-15) causing it to repeat, it means that  \texttt{Head} is changed by another thread executing \textit{enter}, \textit{leave}, or \textit{retire} in the same slot. Thus, that other thread is making progress -- i.e., successfully executing \textit{enter}, \textit{leave}, or \textit{retire} in the same slot, and finishing modification of the same \texttt{Head}. The same argument applies to the second loop (Lines~28-36).
The loop in \textit{traverse} is bounded by the number of batches retired between executing \textit{enter} and \textit{leave}, and
this number is finite.

\hyaline{}-S (Figure~\ref{alg:hyalines}) has two additional loops
(Lines~7-11 and 20-23). If CAS fails in \textit{touch} (Lines~20-23) causing
it to repeat, it means that another thread calling \textit{touch} succeeds.
The other loop (Lines~7-11) converges unless the global era clock is
incremented. In the latter case, another thread is making progress, i.e.,
initializes a new node in Line~17.
\end{proof}

\begin{theorem}
\hyaline{} and \hyaline{}-S have $O(n/k)$ reclamation cost.
\end{theorem}

\begin{proof}
The reclamation cost in Hyaline consists of two parts: 1) the direct cost of \textit{retire} and 2) the cost of \textit{retire} incurred later, during list traversal in \textit{leave}.

Retiring is a simple $O(1)$ linked-list (batch) insertion operation.
Upon reaching the maximum batch size, $s$, \textit{retire} inserts the batch into
slots with active threads. Since the number of slots is $k$, $s\ge k+1$. Each batch is inserted into at most $k$ slots after $s$ per-node
\textit{retire} calls, making the average cost of \textit{retire} $O(1)$.

The list traversal takes places for all batches retired between \textit{enter} and \textit{leave}.
A batch is retired after $s\ge k+1$ \textit{retire} calls (for individual nodes). Each
batch maintains a single reference counter per $\ge (k+1)$ nodes.
The batch's reference counter
needs to be decremented by all active threads (at most, $n$ threads).
Thus, the average cost to update
a reference counter \textit{per node} (i.e., the indirect cost of a single \textit{retire} incurred in \textit{leave}) is $O(n\times \frac{1}{k+1})=O(n/k)$.
\end{proof}

\begin{theorem}
\hyaline{}-1 and \hyaline{}-1S have $O(1)$ reclamation cost.
\end{theorem}

\begin{proof}
\hyaline{}-1 and -1S are special cases of \hyaline{} and \hyaline{}-S, where $k=n$ (i.e., the number of threads equals to the number of slots). Thus, the reclamation cost is $O(1)$.
\end{proof}

\begin{theorem}
\hyaline{}-S and \hyaline{}-1S are robust.\footnote{
As in IBR, we consider only stalled
threads -- i.e.,  threads that are stopped  
indefinitely as opposed to threads that are simply paused briefly.
Also, as in IBR, ``starved'' threads that are running but unable
to make any progress can still potentially reserve an
unbounded number of objects; this is prevented by bounding the number 
of CAS failures in data structure operations and restarting from the very
beginning (implemented by Section~\ref{sec:eval}'s benchmark).}
\label{theor:robust}
\end{theorem}

\begin{proof}
Since slots with stalled threads are detected after
a finite number of retire calls and avoided by active threads in
their following operations, we assume, without
loss of generality, that slots do not reference any
active threads.
(Although batches are potentially added to every slot, one stalled thread can only make unusable the slot which was used by the last \textit{enter} operation of the stalled thread. Only this slot references this thread. Newly allocated nodes will skip this slot due to its stale era and consequently will not reference the stalled thread when these nodes are retired.)

Since threads update their per-slot eras in monotonically increasing
order when calling \textit{deref}, each slot $\mathit{i}$ ends up with
some era $\mathit{A_{i}}$ when it contains only stalled threads.
Let $\mathit{Era_{max}}=\mathit{max}(A_{i})$ across all slots $i$
with stalled threads. We use $\mathit{E_{i}}$ to denote a
global era clock value when the earliest stalled thread from slot $i$ entered.
All previously retired nodes must have been retired before (or at)
$\mathit{E_{i}}$. Let $\mathit{\delta Era}=\mathit{Era_{max}} - \mathit{min}(E_{i})$ across all $\mathit{i}$ slots with stalled threads.
All potentially unreclaimable batches will have their $\mathit{min\_birth}\le \mathit{Era_{max}}$ (Line~14 of Figure~\ref{alg:hyalines}). As each thread periodically increments the era value, a number of unreclaimable
batches is bounded by $\mathit{\delta Era}\times \mathit{Freq}\times n$, where $n$ is the number
of threads and $\mathit{Freq}$ is the frequency used in the algorithm.
Batch sizes can be capped by $k+1$, where
$k$ is the number of slots.
Thus, the number of unreclaimable nodes is bounded by
$\mathit{\delta Era}\times \mathit{Freq}\times n(k+1), k\le n$.
\end{proof}

\section{Evaluation}
\label{sec:eval}

We used and extended the test framework of~\cite{IntervalBased} to support Hyaline.
The framework consists of four benchmarks representing different data structures: the sorted linked-list~\cite{harrisList, hazardPointers}, lock-free hash map~\cite{hazardPointers}, a variant of the Bonsai Tree~\cite{bonsaiTree}, which is a self-balancing lock-free binary tree, and Natarajan and Mittal's binary tree~\cite{natarajanTree}.

We run our tests for up to 144 threads on a 72-core machine consisting
of four Intel Xeon E7-8880~v3 2.30~GHz (45MB~L3 cache) CPUs with hyper-threading disabled and 128GB of RAM. We
chose Clang 11.0.1 with the {\tt -O3} optimization flag due to 
its better support of double-width RMW as used by \hyaline{}.
We saw no visible difference
between GCC and Clang for existing algorithms. We used
jemalloc~\cite{jemalloc} to alleviate the standard library malloc's poor performance~\cite{ALLOCATORS}.

Since a number of different techniques exist, we focus on well-established
or state-of-the-art algorithmic schemes that have similar properties or
programming models as \hyaline{}. We do not evaluate classical reference counting because it uses an intrusive model and is already known to be slower than other evaluated schemes. We skip OS-based approaches since they are inevitably blocking. We skip PEBR~\cite{PEBR} due to
significant API differences. We note that PEBR authors only compare against EBR, and PEBR's performance appears to be 85-90\% of EBR's, worse than that of Hyaline.
Since \hyaline{} aims to achieve excellent throughput while also retaining good memory efficiency, we are comparing against schemes with excellent throughput, such as epoch-based reclamation, and excellent memory efficiency, such as hazard pointers.

We compare all four \hyaline{} variants against:
\begin{description}
\item[\textbf{HP}] -- hazard pointers~\cite{hazardPointers}.
\item[\textbf{HE}] -- hazard eras~\cite{hazardEra}.
\item[\textbf{IBR}] -- the interval-based technique 2GE-IBR~\cite{IntervalBased}.
\item[\textbf{Epoch}] -- a variant~\cite{IntervalBased} of the epoch-based approach.\footnote{This variant has an advantage over the original EBR~\cite{epoch1,epoch2} in that it increments the epoch counter unconditionally, but it has to place all retired nodes in one per-thread list. Both approaches exhibit good performance.}
\item[\textbf{No MM}] -- running the test without any memory reclamation, which serves as a general baseline.
\end{description}

In the results, it is more fair to compare \hyaline{} and \hyaline{}-1 against (non-robust) Epoch, and \hyaline{}-S and \hyaline{}-1S against (robust) HP, IBR, and HE.

The original benchmark code we used~\cite{IntervalBased} implemented snapshots only for IBR. HP and HE were suboptimal due to excessive cache misses when scanning lists of retired nodes. We modified these implementations accordingly. EBR and Hyaline are snapshot-free and do not need this
optimization.

\begin{figure*}[ht]
\begin{subfigure}{.333\textwidth}
\includegraphics[width=.99\textwidth]{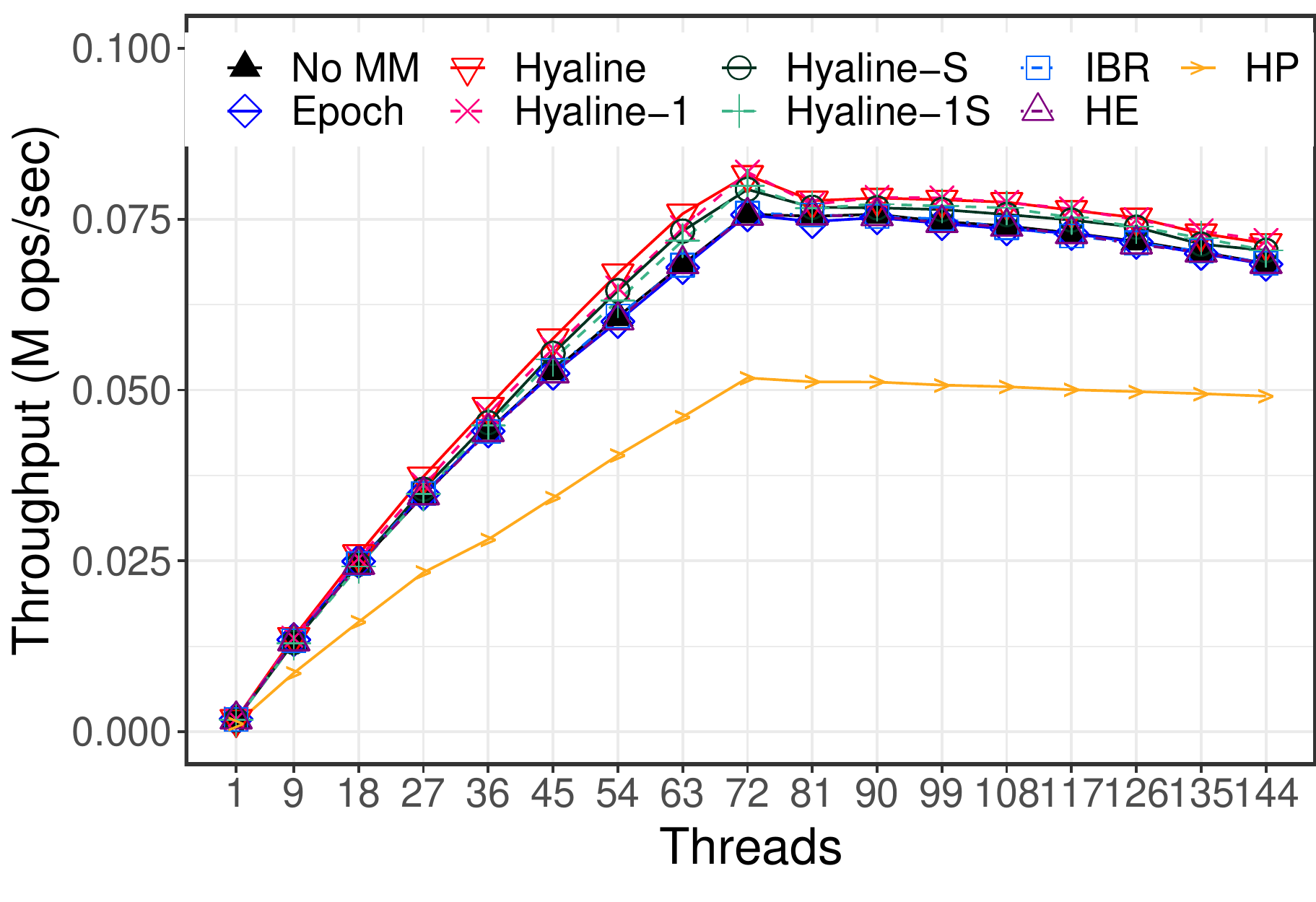}
\vspace{-5pt}
\caption{Harris \& Michael list (write)}
\label{fig:list_thru}
\end{subfigure}%
\begin{subfigure}{.333\textwidth}
\includegraphics[width=.99\textwidth]{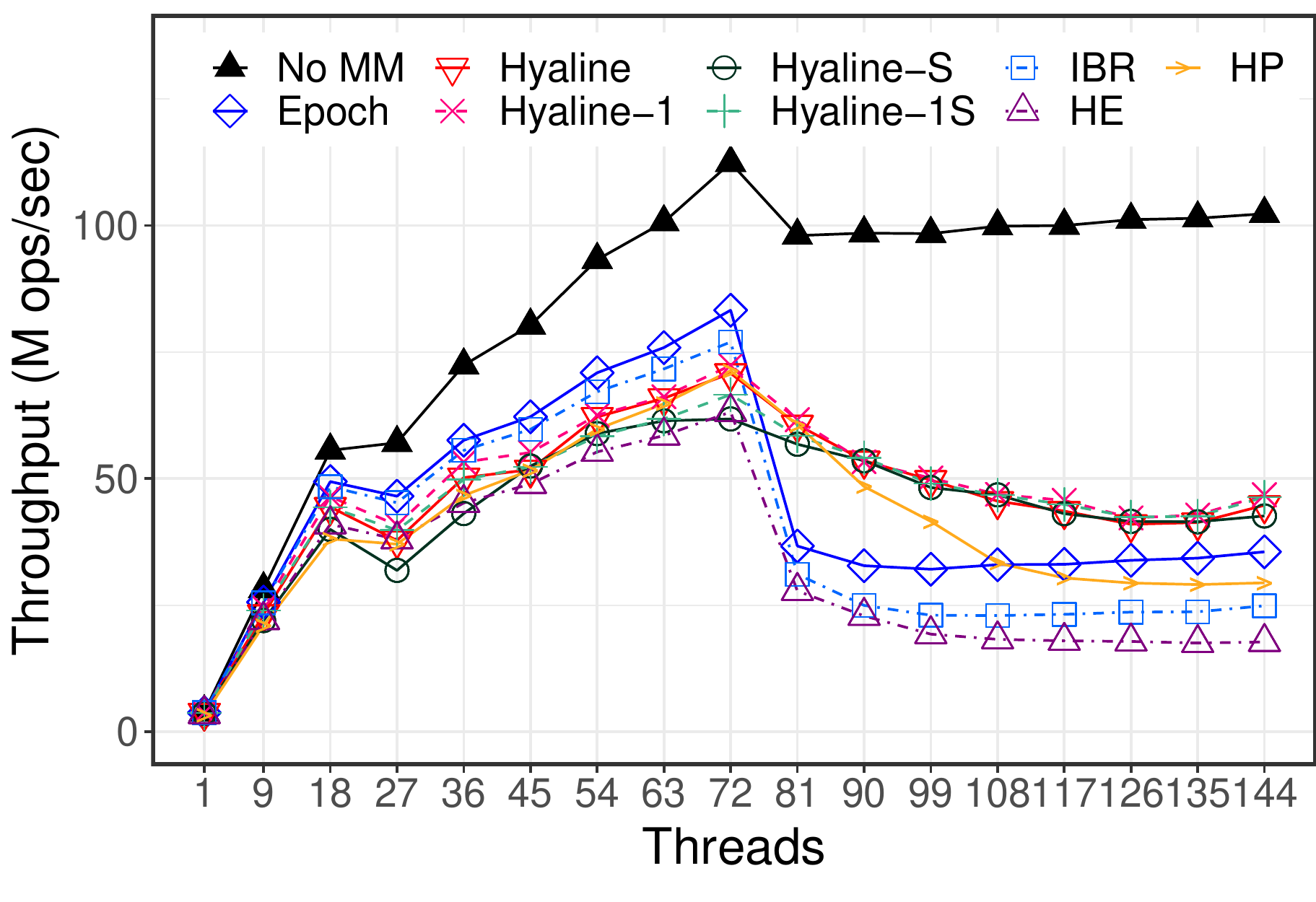}
\vspace{-5pt}
\caption{Michael hash map (write)}
\label{fig:hash_thru}
\end{subfigure}%
\begin{subfigure}{.333\textwidth}
\includegraphics[width=.99\textwidth]{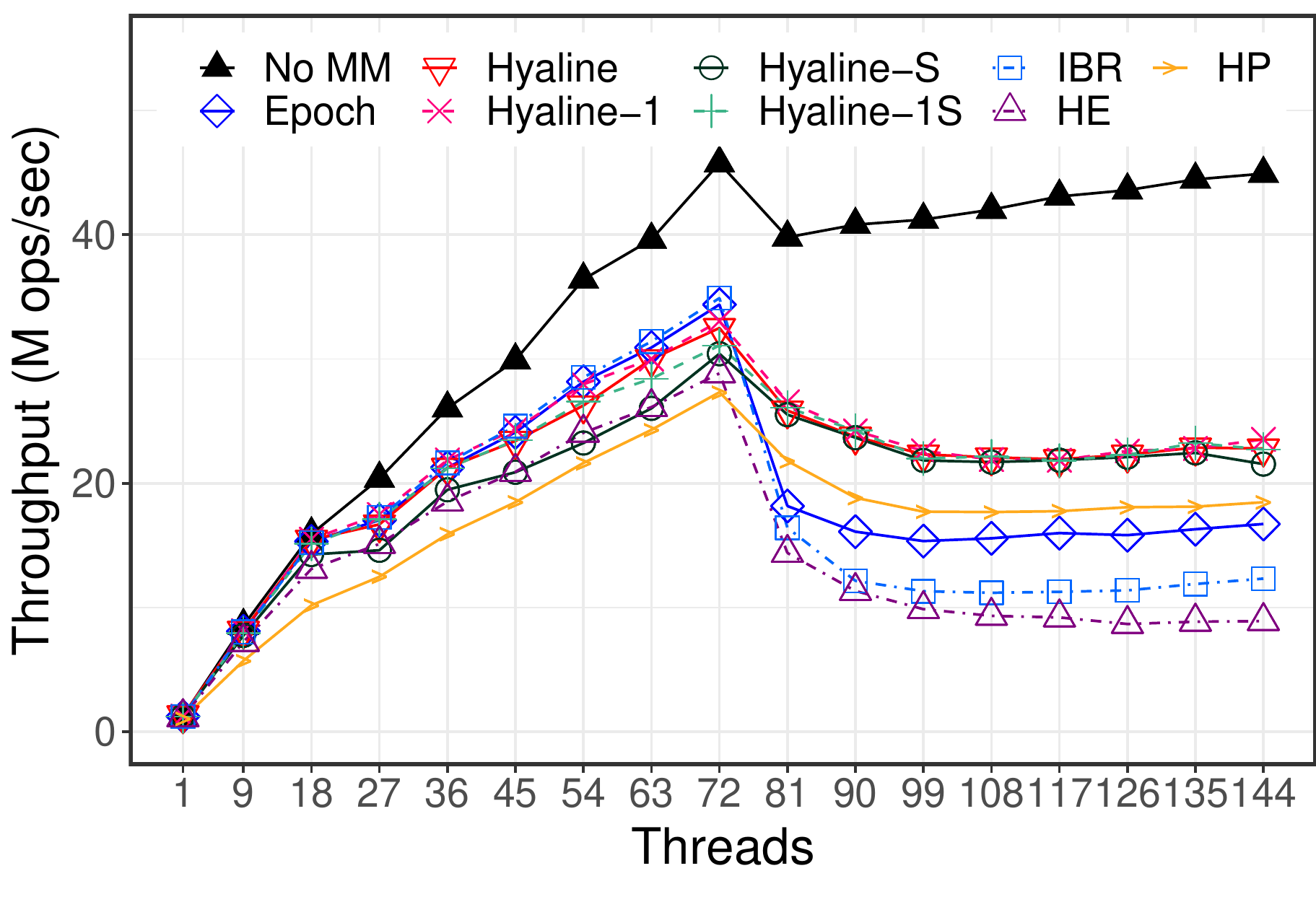}
\vspace{-5pt}
\caption{Natarajan \& Mittal tree (write)}
\label{fig:natarajan_thru}
\end{subfigure}%
\\
\begin{subfigure}{.333\textwidth}
\includegraphics[width=.99\textwidth]{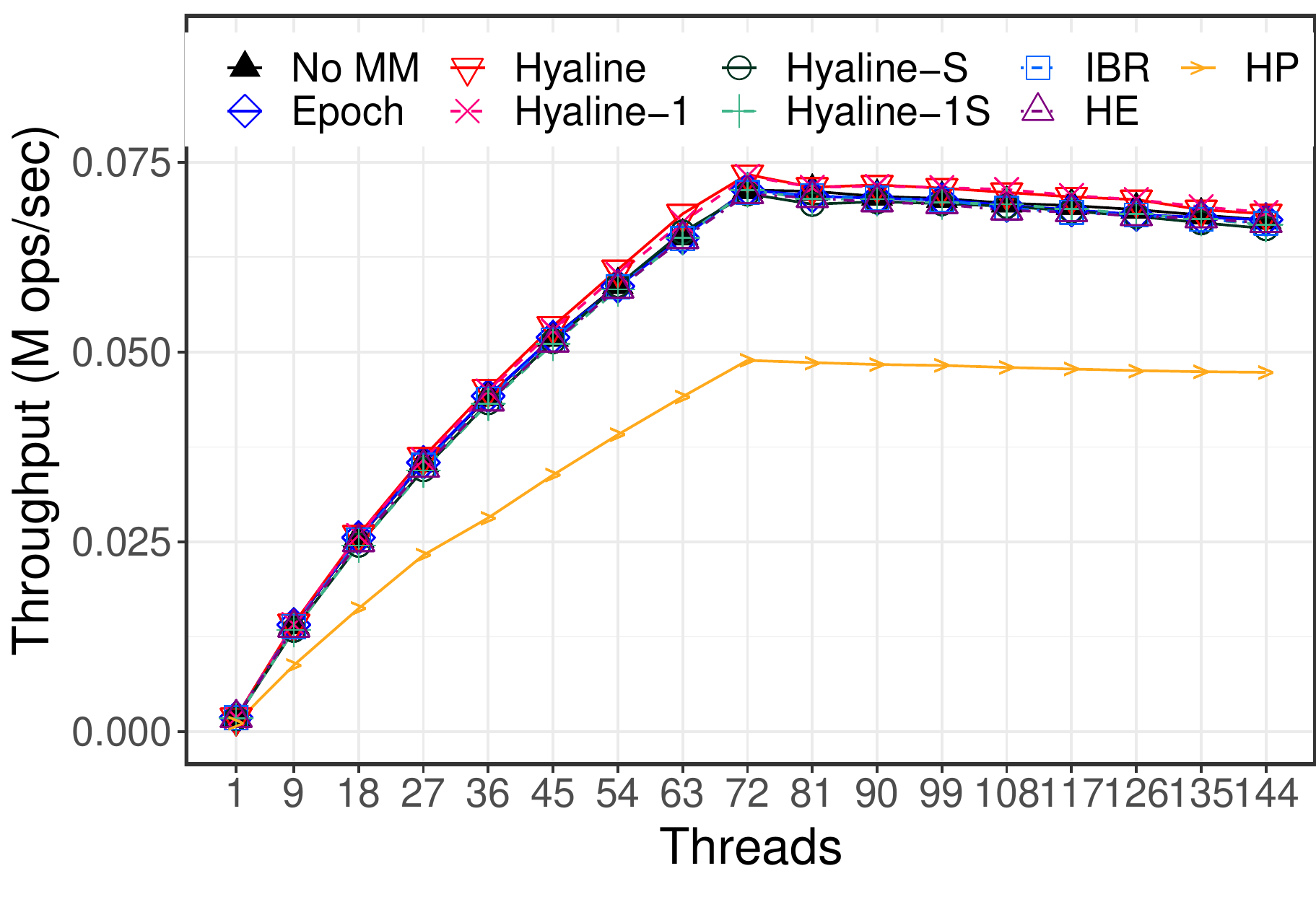}
\vspace{-5pt}
\caption{Harris \& Michael list (read)}
\label{fig:list_thru_read}
\end{subfigure}%
\begin{subfigure}{.333\textwidth}
\includegraphics[width=.99\textwidth]{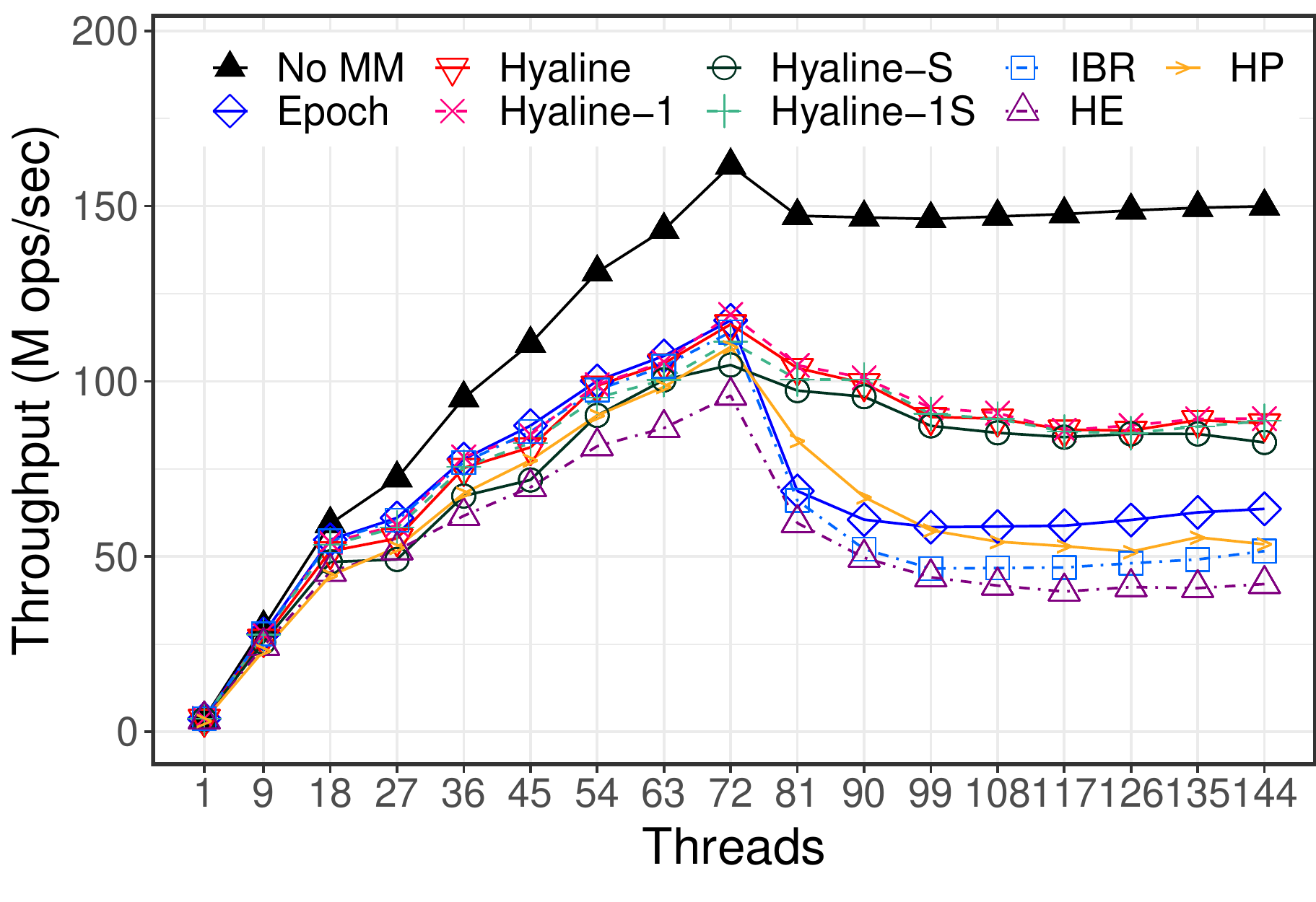}
\vspace{-5pt}
\caption{Michael hash map (read)}
\label{fig:hash_thru_read}
\end{subfigure}%
\begin{subfigure}{.333\textwidth}
\includegraphics[width=.99\textwidth]{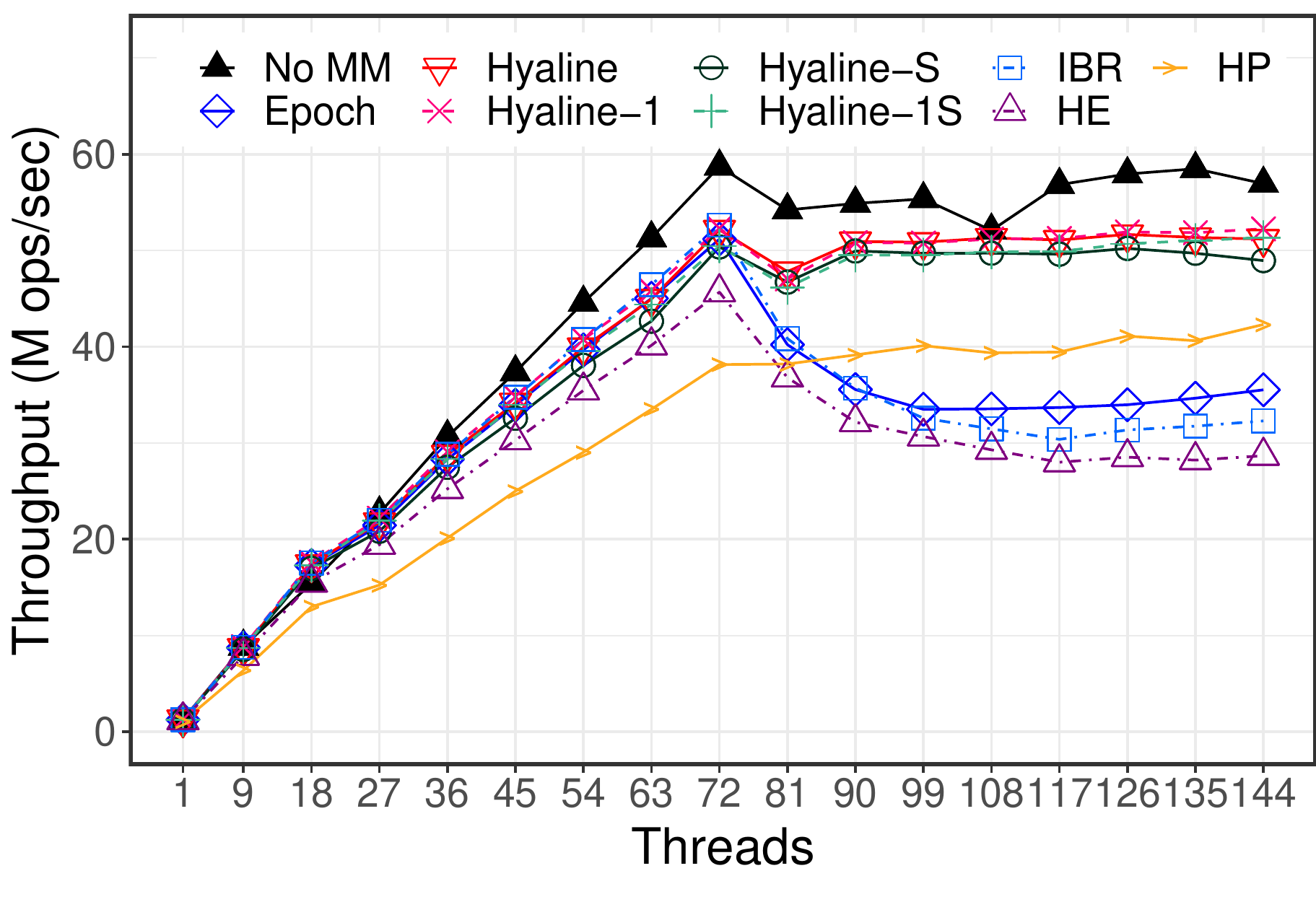}
\vspace{-5pt}
\caption{Natarajan \& Mittal tree (read)}
\label{fig:natarajan_thru_read}
\end{subfigure}%
\caption{Throughput (higher is better).}
\label{fig:thru}
\end{figure*}

\begin{figure*}[ht]
\begin{subfigure}{.333\textwidth}
\includegraphics[width=.99\textwidth]{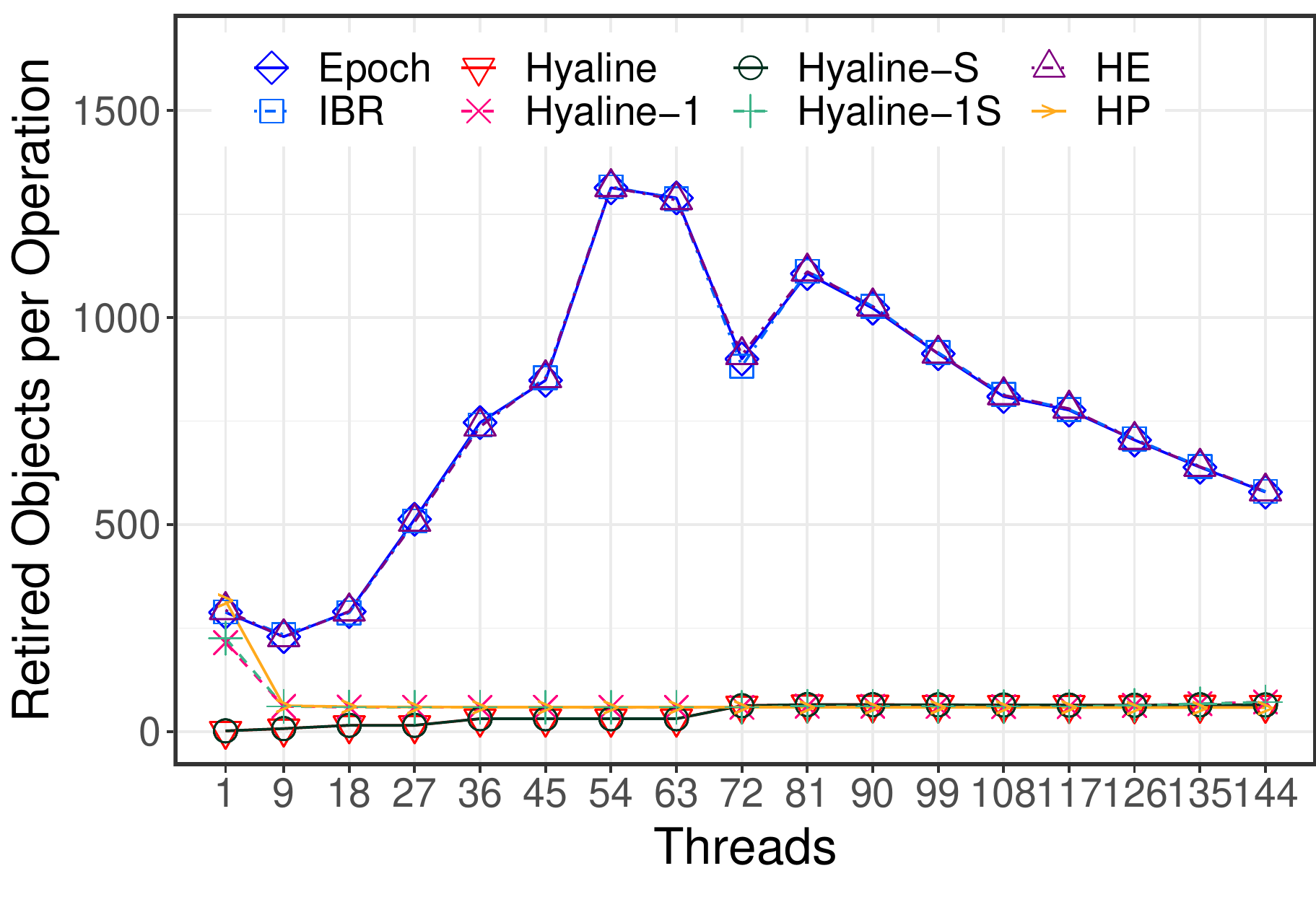}
\vspace{-5pt}
\caption{Harris \& Michael list (write)}
\label{fig:list_unrec}
\end{subfigure}%
\begin{subfigure}{.333\textwidth}
\includegraphics[width=.99\textwidth]{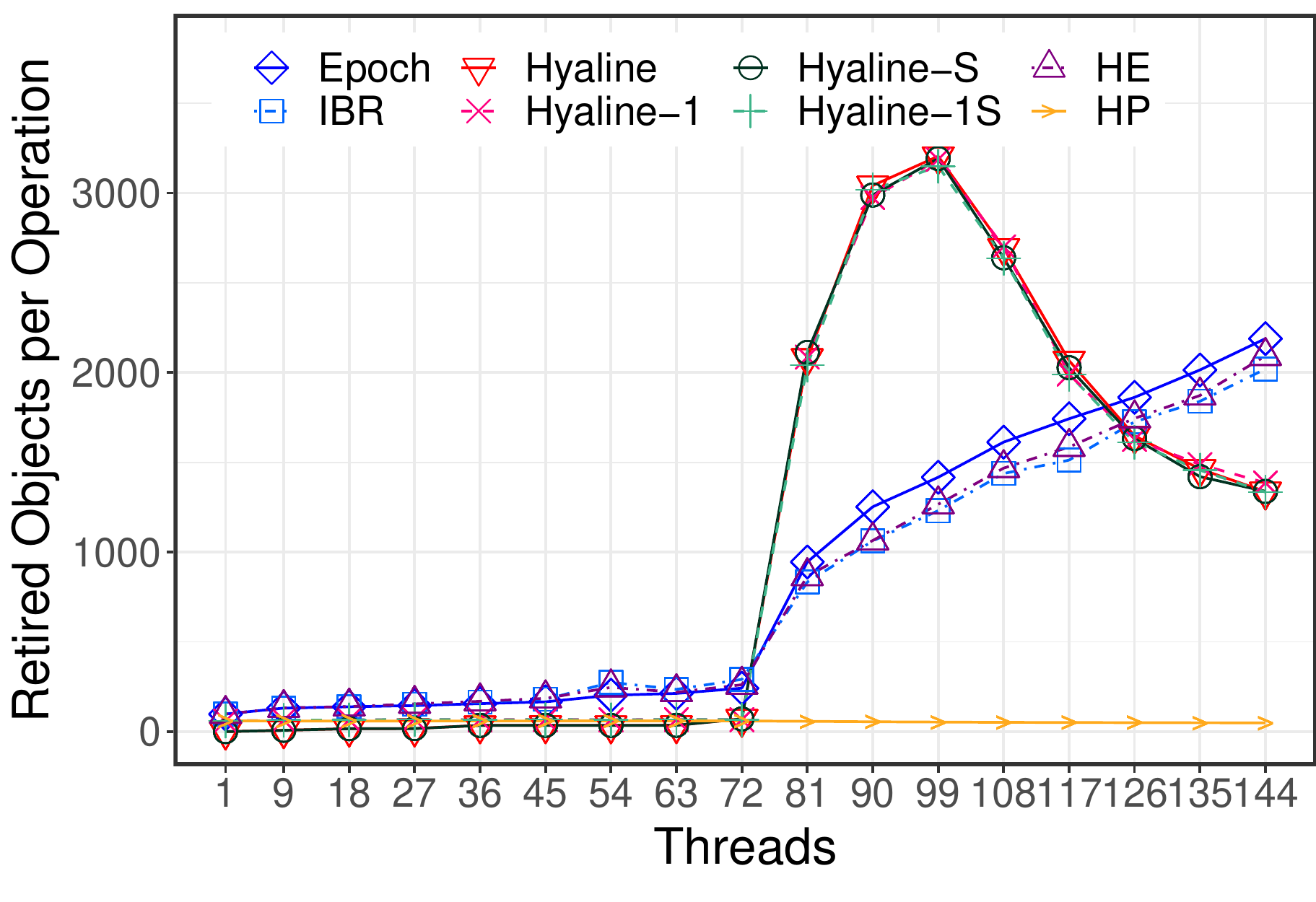}
\vspace{-5pt}
\caption{Michael hash map (write)}
\label{fig:hash_unrec}
\end{subfigure}%
\begin{subfigure}{.333\textwidth}
\includegraphics[width=.99\textwidth]{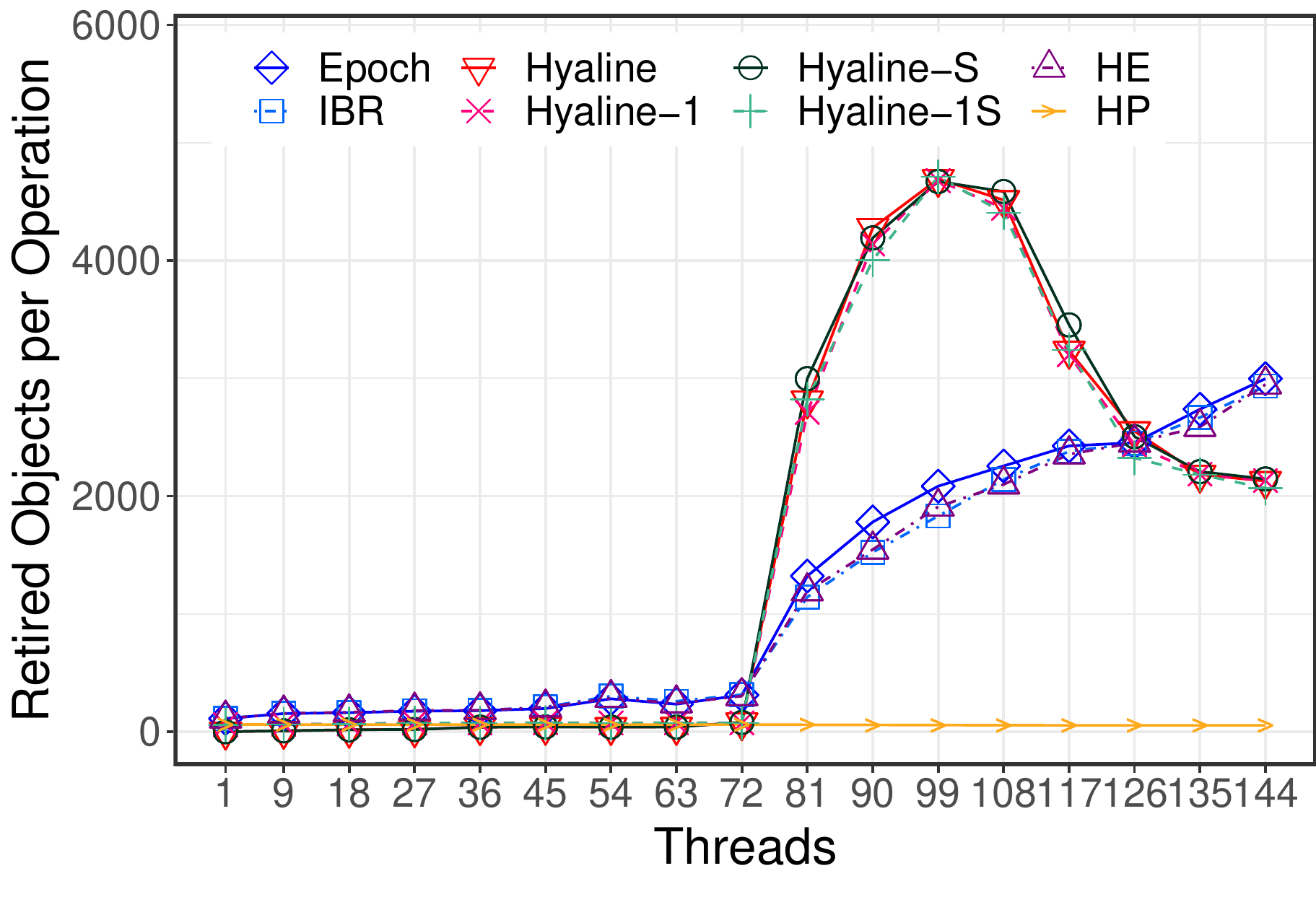}
\vspace{-5pt}
\caption{Natarajan \& Mittal tree (write)}
\label{fig:natarajan_unrec}
\end{subfigure}%
\\
\begin{subfigure}{.333\textwidth}
\includegraphics[width=.99\textwidth]{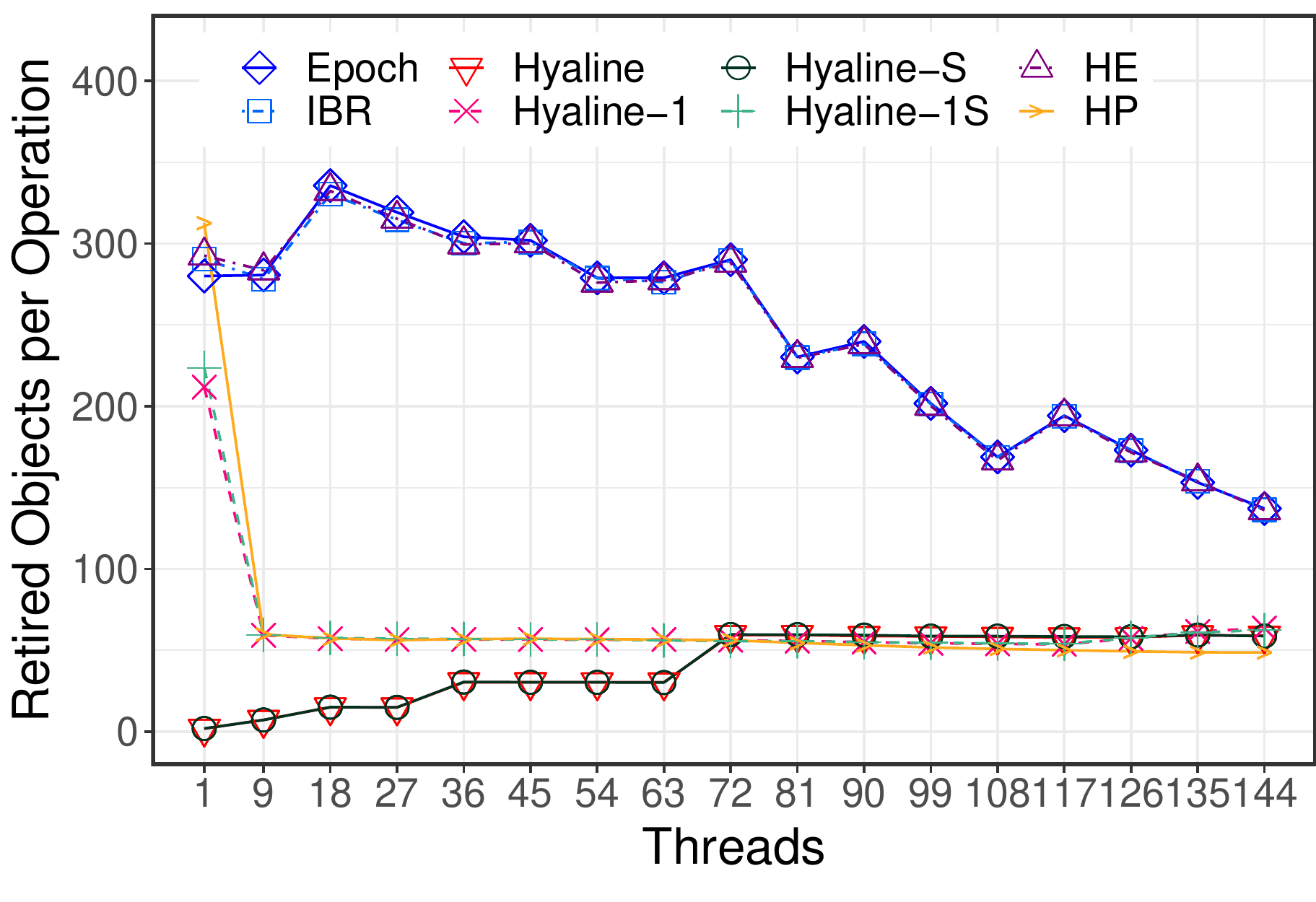}
\vspace{-5pt}
\caption{Harris \& Michael list (read)}
\label{fig:list_unrec_read}
\end{subfigure}%
\begin{subfigure}{.333\textwidth}
\includegraphics[width=.99\textwidth]{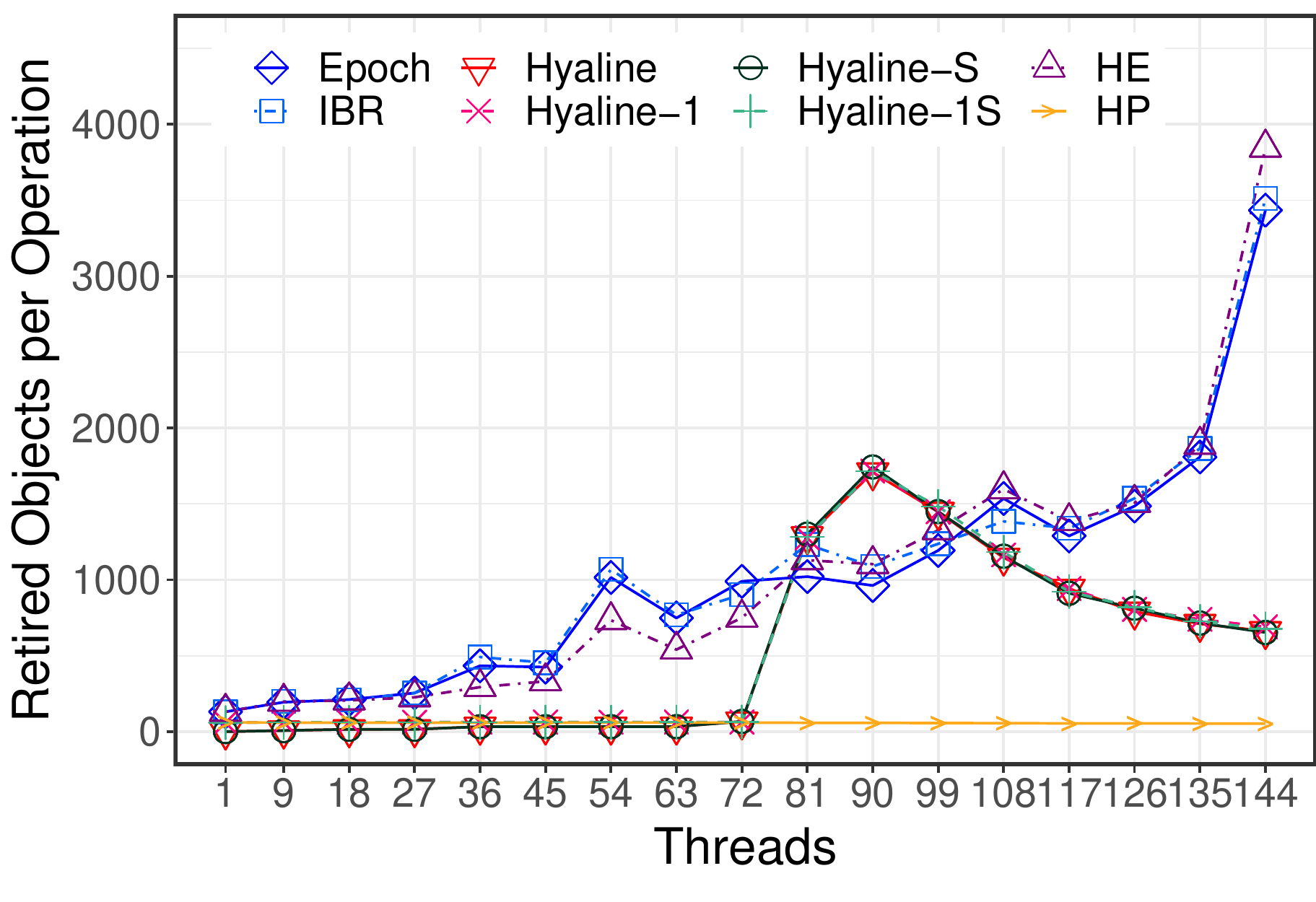}
\vspace{-5pt}
\caption{Michael hash map (read)}
\label{fig:hash_unrec_read}
\end{subfigure}%
\begin{subfigure}{.333\textwidth}
\includegraphics[width=.99\textwidth]{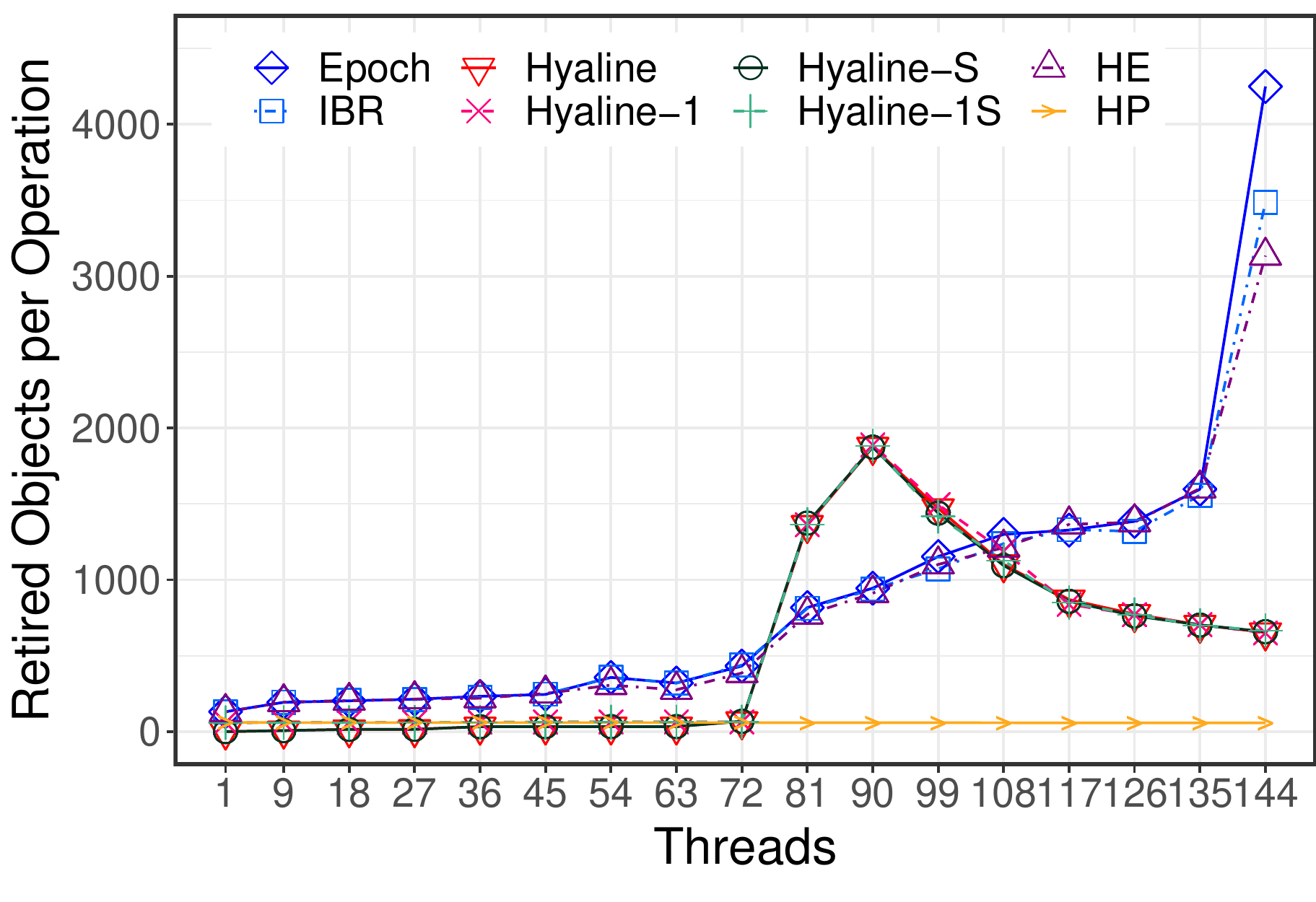}
\vspace{-5pt}
\caption{Natarajan \& Mittal tree (read)}
\label{fig:natarajan_unrec_read}
\end{subfigure}%
\caption{Average number of unreclaimed objects per operation (lower is better).}
\label{fig:unrec}
\end{figure*}

Note that the actual throughput can exceed \textit{No MM} as it can be faster to recycle old objects. As memory deallocation slows down due to a number of factors, including number of freed objects, any memory reclamation scheme can also become objectively slower than \textit{No MM}.

We use both a \textit{write}-intensive workload (50\% \textit{insert}, 50\% \textit{delete}), which stresses reclamation techniques through a large number of insertions and deletions, as well as \textit{read}-dominated workload (90\% \textit{get}, 10\% \textit{put}), which represents a more reclamation-unbalanced and yet common scenario.

For each data point, the experiment starts by prefilling the data structure with 50,000 elements and runs 10 seconds.
Each thread then randomly performs the aforementioned operations. The key used in each operation is randomly chosen from the range of 0 to 100,000 with equal probability. We run the experiment 5 times and report the average.

\begin{figure*}[ht]
\begin{subfigure}{.333\textwidth}
\includegraphics[width=.99\textwidth]{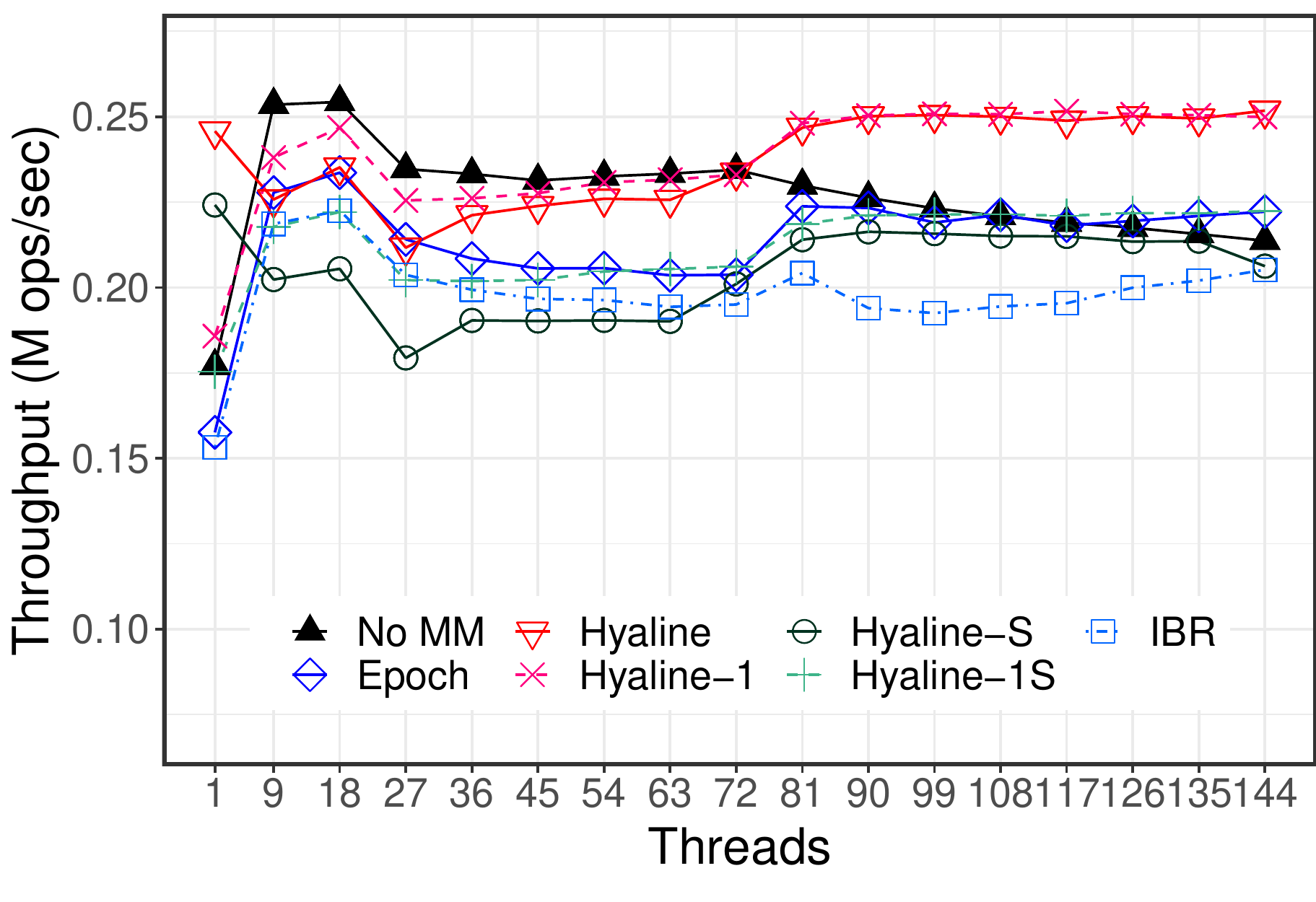}
\vspace{-5pt}
\caption{Throughput (write)}
\label{fig:bonsai_thru}
\end{subfigure}%
\begin{subfigure}{.333\textwidth}
\includegraphics[width=.99\textwidth]{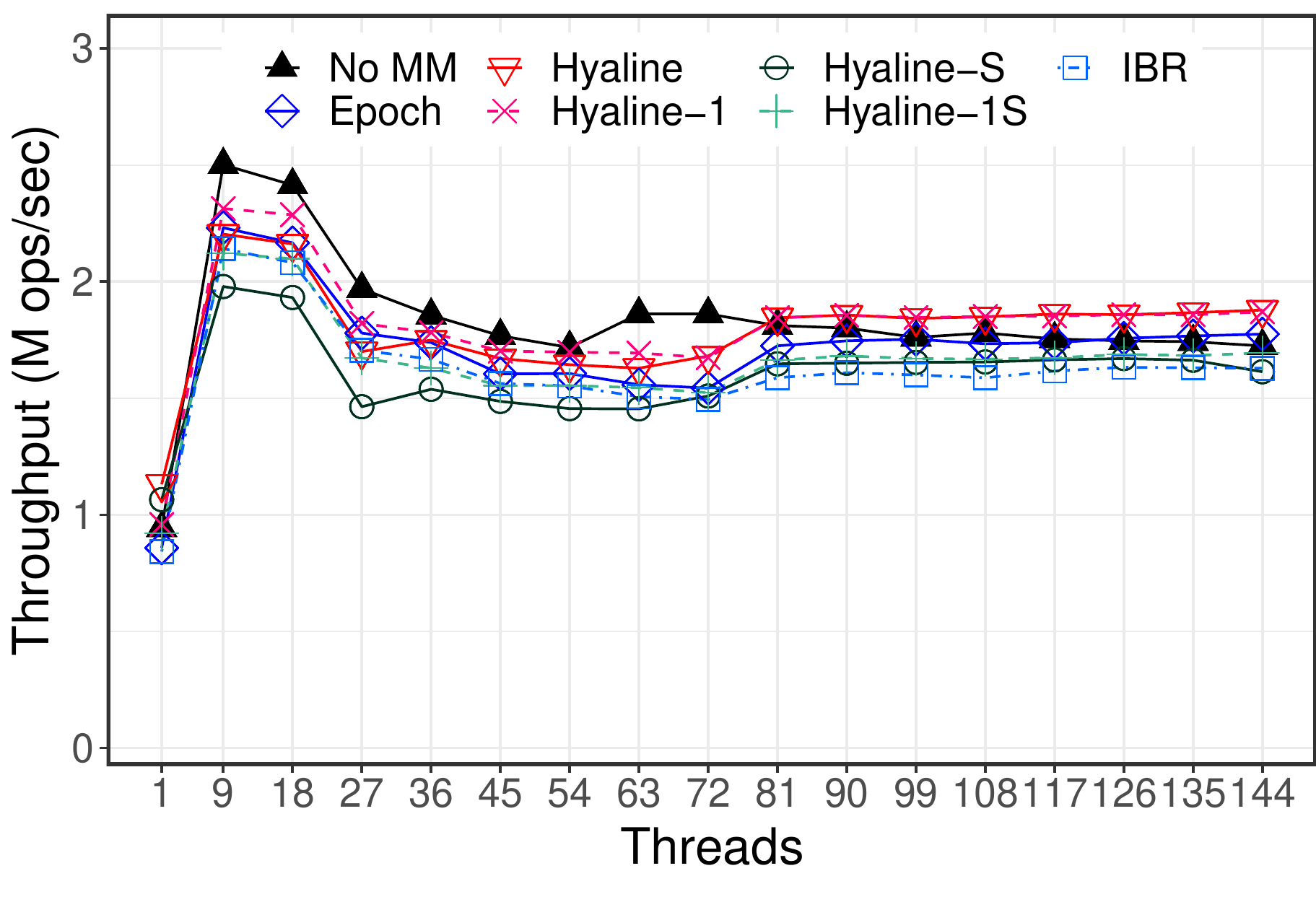}
\vspace{-5pt}
\caption{Throughput (read)}
\label{fig:bonsai_thru_read}
\end{subfigure}%
\begin{subfigure}{.333\textwidth}
\includegraphics[width=.99\textwidth]{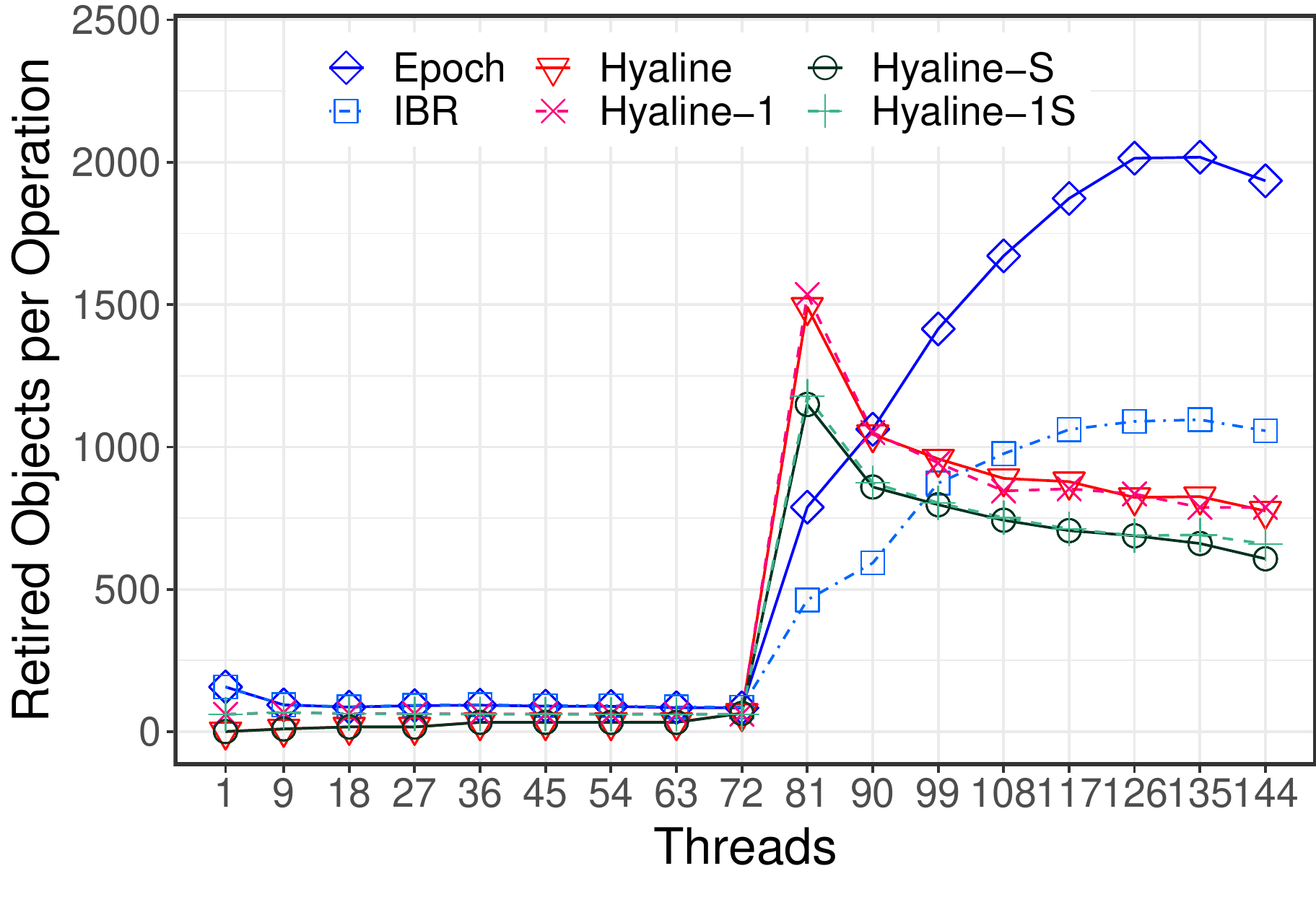}
\vspace{-5pt}
\caption{Unreclaimed objects (write)}
\label{fig:bonsai_unrec}
\end{subfigure}%
\caption{Bonsai tree. (Unreclaimed objects for \textbf{read} and \textbf{write} are nearly identical.)}
\label{fig:bonsai}
\end{figure*}

Reclamation algorithms need to be adjusted to gain good performance.
Although this process is tricky, we found more or less reasonable
parameters for a fair comparison such that existing algorithms achieve the highest possible throughput while retaining as much of memory efficiency as possible.
For our machine, benchmark parameters $\mathit{epochf}=150$ and $\mathit{emptyf}=120$ appear to be optimal for existing schemes in this regard.
$\mathit{epochf}$ amortizes the frequency of
epoch counter increments for Epoch, IBR, and HE. $\mathit{emptyf}$ reduces other
overheads for all algorithms, e.g., amortizing the frequency of list traversals.
For \hyaline{}
and \hyaline{}-S, we cap the number of slots, $k$, at $128$ (the next power of 2 of the number of cores). 
All variants use batches of at least 64 and at most $k+1$ nodes (as required by
the \hyaline{} algorithms).

Figure~\ref{fig:list_thru} shows the throughput of (sorted) Linked-list, which is a good example of an unbalanced workload since 
operations are slow and dominated by  the long traversal required to find an element (even in the write-intensive scenario).
Figure~\ref{fig:list_unrec}, which shows the average number of retired but not-yet-reclaimed
objects per operation (allows us to estimate how fast memory is reclaimed), demonstrates that \hyaline{} has excellent memory efficiency, which is much better than that of Epoch, HE, or IBR.
This validates our claim that \hyaline{}'s efficiency is
better in unbalanced settings.
All \hyaline{} variants also have marginally higher throughput than the other schemes. Although HP is also efficient, its throughput is visibly worse due to so many memory barriers incurred while traversing the list. Similar trends are also observed for the read-dominated case (Figures~\ref{fig:list_thru_read} and~\ref{fig:list_unrec_read}).

Figure~\ref{fig:hash_thru} shows hash map's throughput using the write-intensive workload. Hash map operations are very short and significantly stress memory reclamation systems. Because operations are short, HP's performance does
not degrade as much as in the prior test.
The gap between \textit{No MM} and memory reclamation techniques substantially
increases as the number of threads begins to exceed the number of cores. However, all \hyaline{}
variants still perform well after 72 threads. (The gap between \hyaline{}
and Epoch gets as large as \textbf{2x} for 81 threads.) For a smaller number
of threads, retirement in \hyaline{} can be slightly more expensive than in Epoch and IBR.
The average number of unreclaimed objects (Figure~\ref{fig:hash_unrec}) for all \hyaline{} variants
is comparable to HP and smaller than that of IBR, HE, or Epoch before the oversubscribed
scenario (not visible due to a smaller scale). Although it temporarily increases afterwards,
the corresponding throughput is also substantially higher than that of other schemes. Hence, one possible explanation for this increase is that \hyaline{} simply allocates and reclaims more objects (compared to other schemes) in the first place.
Since this workload is already very balanced, \hyaline{} also does
not get any extra benefit due to reclamation balancing.
Hash map's results are somewhat more interesting for the read-dominated case (Figures~\ref{fig:hash_thru_read} and~\ref{fig:hash_unrec_read}), where \hyaline{} is more memory efficient than IBR, HE, or Epoch.
\hyaline{}'s throughput remains very high, even in oversubscribed scenarios.

Natarajan \& Mittal tree (Figures~\ref{fig:natarajan_thru},~\ref{fig:natarajan_unrec},~\ref{fig:natarajan_thru_read}, and~\ref{fig:natarajan_unrec_read}) shows
similar trends to that of hash map. HP is slower due to longer operations. Throughput gains of \hyaline{} are more visible here.
With respect to memory efficiency, we see the same benefit in the read-dominated workload as in hash map.
Before oversubscription, \hyaline{}'s efficiency is close to HP's.

Figures~\ref{fig:bonsai_thru} and~\ref{fig:bonsai_thru_read} show Bonsai tree's throughput. HP and HE are not implemented due to
the complexity of the tree rotation operations~\cite{IntervalBased}, for which
the number of local pointers cannot be determined in advance.
Throughput drops for all schemes as we approach 18 per-socket cores, most
likely due to over-socket contention~\cite{IntervalBased}.
\hyaline{} and \hyaline{}-1 achieve the best performance and steadily outperform Epoch by $\approx$10\%.
All robust schemes presented for this benchmark (IBR, \hyaline{}-S, 
and \hyaline{}-1S) have similar performance; it is worse than their non-robust
counterparts due to increased number of pointer dereferences. The
number of unreclaimed objects (Figures~\ref{fig:bonsai_unrec}) for \hyaline{} and \hyaline{}-S is mostly smaller than that of Epoch and IBR, respectively.

\subsubsection*{Snapshots}
Snapshots also impact memory utilization. For 144 threads and 16 local (concurrently reserved) pointers in HP and HE, snapshots \textbf{additionally require 2.5 MB}. Although this size depends on the number of threads, and the number of local pointers is typically smaller, we present this figure to give some perspective. For example, even if the number of unreclaimed objects is as high as $4000$, and each object is 128 bytes, we still use less than \textbf{0.5 MB}. To retain the same evaluation methodology as in prior works, our results above disregard snapshot overheads. However, for snapshot-based schemes (IBR, HP, and HE), snapshots alone
can create more memory inefficiency than the scheme itself, a fact rarely acknowledged in prior works.

\section{Related Work}
\label{sec:related}

A number of approaches for safe memory reclamation (SMR) were proposed over
the last two decades.

Most SMR approaches are either pointer- or epoch- based. Pointer-based
techniques such as hazard pointers (HP)~\cite{hazardPointers} are typically
fine-grained and track every accessed object.
Unfortunately, this approach degrades 
performance as pointer dereferencing incurs additional overheads, such
as memory writes and barriers.
Pass-the-buck~\cite{Herlihy:2002:ROP:645959.676129,Herlihy:2005:NMM:1062247.1062249} has a similar model.
Drop the anchor~\cite{Braginsky:2013:DAL:2486159.2486184} is designed specifically for linked-lists and outperforms hazard pointers, but the approach is not directly applicable to other data structures.
Optimistic Access~\cite{Cohen:2015:EMM:2755573.2755579} and Automatic Optimistic Access (AOA)~\cite{Cohen:2015:AMR:2814270.2814298} are more universal techniques, but they require  data structures to be written in a ``normalized form.'' FreeAccess~\cite{Cohen:2018:DSD:3288538.3276513} drops this requirement and implements a garbage collector. FreeAccess, however,
needs to divide a program
into read-only and write-only periods, which makes it impossible to directly use certain operations such as \textit{swap}.
OrcGC~\cite{OrcGC}, another fully lock-free garbage collector, achieves
good performance but is still slower in some tests than HP. 

In epoch-based reclamation (EBR)~\cite{epoch1,epoch2}, which is based on the read-copy-update (RCU)~\cite{Mckenney01read-copyupdate} paradigm, objects are marked with the current
epoch value at the time they are retired.
A memory object is deallocated only when all thread
reservations are ahead of the object's retire epoch and no thread
can reach it. Stamp-it~\cite{Poter:2018:BAS:3210377.3210661} extends EBR to guarantee $O(1)$ reclamation cost but is not robust
and requires per-thread control blocks. It
extends EBR by using a
doubly-linked list, and requires
ABA tags~\cite{Herlihy:2008:AMP:1734069}.
Stamp-it squeezes 17-bit tags directly into control block pointers, but for ABA safety, it is better to use
larger tags and double-width CAS.

The hazard eras (HE) approach~\cite{hazardEra} attempts to reconcile
EBR with HP: HE is robust, but uses ``eras''
(i.e., epochs) instead of pointer addresses to accelerate the algorithm.
When allocating memory objects, they are tagged with the \textit{birth} era, and
when objects are retired, they are tagged with the \textit{retire} era.
Lifecycles of objects are controlled by these eras. Similarly to HP's API model, indices must be assigned to all
accessed objects in HE.
A subsequent work~\cite{WFE} makes HE wait-free.
Interval-based reclamation (IBR)~\cite{IntervalBased} employs the idea of
birth and retire eras but forgoes the need to explicitly assign indices making its API model, especially in its 2GE-IBR variant, reminiscent of EBR and easier to use.

Some approaches exploit OS support.
PEBR~\cite{PEBR} relies on OS tricks~\cite{OSTRICKS} to avoid extra memory barriers, which makes it intrusive to execution environments.
DEBRA+~\cite{debra}, NBR~\cite{NBR}, and QSense~\cite{Balmau:2016:FRM:2935764.2935790} improve EBR to make it
robust, but they rely on OS signals or scheduler support.
They are robust but not in a fully lock-free manner as typical OSs such
as Linux inevitably use locks. ThreadScan~\cite{threadScan} and Forkscan~\cite{forkScan} are other examples of schemes that rely on signals.

Another approach that is simple to implement but has a high overhead is lock-free reference counting (LFRC)~\cite{refCount,Valois:1995:LLL:224964.224988}. In this approach, each object is associated with a reference count. An object can be safely reclaimed when the reference count reaches zero. The reference count is updated with every access, which converts read-accesses into write-accesses with a memory barrier. This significantly impacts performance.
\hyaline{} uses a completely different approach, wherein
objects are accessed without modifying reference counters. Since
active threads are tracked only in the list of \emph{retired} objects,
\hyaline{}'s overhead is significantly smaller.

Some approaches rely on hardware transactional memory (HTM) to speed up reference counting~\cite{htmRefCount} using HTM transactions. Another approach~\cite{readAsHTM} executes any read operation on the data structure as an HTM transaction. When a conflict occurs in a concurrent
thread that reclaims memory, the transaction is aborted. Some other approaches~\cite{pageFaultRec} optimize performance by using page protection mechanisms which issue a page fault that forces a global memory barrier.

\section{Conclusion}

We presented \hyaline{}, a new lock-free algorithm for
safe memory reclamation. \hyaline{} uses LL/SC or
double-width CAS, which are available on most modern architectures. A specialized
\hyaline{}-1 algorithm uses single-width CAS and can be implemented on all architectures. We also presented  \hyaline{}-S and \hyaline{}-1S extensions,
which bound memory usage even in the presence of stalled threads.
Compared to other common approaches, the \hyaline{} schemes balance the reclamation workload due to their underlying asynchronous nature of reclamation. This often manifests in improved memory
efficiency without sacrificing performance.

All \hyaline{} schemes are
suitable for environments where threads are created and
deleted dynamically:
threads are ``off-the-hook'' as soon as they \textit{leave}
and do not need to check retirement lists afterwards. 
\hyaline{} and \hyaline{}-S are fully transparent as
they need not explicitly register or unregister threads; they
can allocate a fixed number of slots
roughly corresponding to the number of cores and still support
any number of threads.
\hyaline{}-1 and \hyaline{}-1S are less transparent in this sense but can be implemented everywhere.

\hyaline{} schemes do not take snapshots, which can help reduce memory footprints as the number of threads grow.

We tested all \hyaline{} versions on x86(-64),
ARM32/64, PowerPC, and MIPS architectures. For these architectures, all
\hyaline{} variants exhibit very high throughput on various data structures, and ensure that the number of retired, but not-yet-reclaimed objects is small. 
We presented results for x86-64, a ubiquitous architecture.
\hyaline{}'s benefits are especially visible in certain read-dominated
workloads. Moreover, in oversubscribed scenarios, \hyaline{} obtains up to \textbf{2x}
throughput gain over other algorithms, including EBR.

\section*{Availability}
We provide code for the modified benchmark and all Hyaline variants at \url{https://github.com/rusnikola/lfsmr}.

The arXiv version of the paper is available at \url{https://arxiv.org/abs/1905.07903}.

\section*{Acknowledgments}
A preliminary version of the algorithm previously appeared as a brief announcement at PODC '19~\cite{hyalineBA}.

We would like to thank the anonymous reviewers and our
shepherd Tony Hosking for their insightful comments and
suggestions, which helped greatly improve the paper.
We also thank Mohamed Mohamedin for helping with experiments for an early
version of the algorithm.

This work is supported in part by AFOSR under grants FA9550-15-1-0098
and FA9550-16-1-0371 and ONR under grants N00014-18-1-2022 and
N00014-19-1-2493.

\bibliography{lockfree}

\appendix
\newcommand{\llscone}{%
\begin{algorithm2e}[H]
\setcounter{AlgoLine}{10}
\Fn{\textbf{head\_t} dFAA(\textbf{head\_t} *Head, \textbf{int} RefAddend)} {
\tcp{Increment HRef}
\tcp{HPtr remains intact}
\Do {\Not SC(Head.HRef, Value)}{
Old.HRef = LL(Head.HRef);\

Old.HPtr = Load(Head.HPtr);\

Value = Old.HRef + RefAddend;\
}
\Return Old;\
}
\end{algorithm2e}
}

\newcommand{\llsctwo}{%
\begin{algorithm2e}[H]
\Fn{\textbf{bool} dCAS\_Ptr(\textbf{head\_t} *Head, \textbf{head\_t} Expect, \textbf{head\_t} New)} {
Old.HPtr = LL(Head.HPtr);\

Old.HRef = Load(Head.HRef);\

\lIf{Old $\neq$ Expect} {\Return False}
\Return SC(Head.HPtr, New.HPtr);\
}
\Fn{\textbf{bool} dCAS\_Ref(\textbf{head\_t} *Head, \textbf{head\_t} Expect, \textbf{head\_t} New)} {
Old.HRef = LL(Head.HRef);\

Old.HPtr = Load(Head.HPtr);\

\lIf{Old $\neq$ Expect} {\Return False}
\Return SC(Head.HRef, New.HRef);\
}
\end{algorithm2e}
}

\begin{figure*}
\begin{subfigure}{.5\textwidth}
\llsctwo
\end{subfigure}%
\begin{subfigure}{.5\textwidth}
\llscone
\end{subfigure}%
\vspace{-8pt}
\caption{Single-width LL/SC version. Both HRef and HPtr are in the same reservation granule.}
\label{alg:llsc}
\end{figure*}

\section{\hyaline{}(-S) for Single-width LL/SC}
\label{sec:llsc}

ARM32 MPCore+ and ARM64~\cite{ARM:manual} implement
double-width LL/SC. \hyaline{} can be easily 
implemented on such architectures. However, MIPS~\cite{MIPS:manual}
(until very recently) and PowerPC~\cite{PPC:manual} implement only single-width
LL/SC. We now describe an approach that we have used to successfully implement (in inline assembly) and stress-test
\hyaline{}.

Typical implementations of LL/SC have certain restrictions, e.g., memory
writes between LL and SC are prohibited. Memory reads, however, are still allowed for architectures such as PowerPC and MIPS.
Although architectural manuals do not always mention it explicitly,
the LL reservation granularity is typically larger than just a single CPU
word, i.e., an entire L1 cache line~\cite{Sarkar:2012:SCP:2254064.2254102} or even larger.
This creates ``false sharing,'' where concurrent LL/SC on
unrelated variables residing in the same granule causes SC to
succeed only for one variable.
In \hyaline{}, the {\tt [HRef,HPtr]} tuple needs
to be atomically loaded. However, with the exception of a special case
when {\tt HRef} goes to $0$ in \textit{leave}, we only update 
one or the other variable at a time. We use this observation to implement
\hyaline{} for single-width LL/SC. We place two variables in the same
reservation granule by aligning the first variable on the double-word
boundary so that only one LL/SC
pair succeeds at a time. An ordinary \textit{load} operation between
LL and SC loads the other variable.
To prevent reordering of LL and load, we construct an
artificial data dependency that acts as an implicit barrier for load.
For the SC to
succeed, the other variable from the granule must also be intact.

In Figure~\ref{alg:llsc}, we present a replacement of the
\textit{FAA} operation used by \textit{enter} and two versions of double-width CAS
replacements. \textit{retire} uses the version that modifies a pointer.
\textit{leave} first uses the version that modifies a reference counter.
We keep {\tt HPtr} intact in \textit{leave} even if {\tt HRef}
goes to $0$. Then the other version of CAS sets {\tt HPtr} to
{\tt Null} if the object is still unclaimed by any concurrent
\textit{enter}. Double-width load atomicity is guaranteed only when
SC succeeds. Our algorithm tolerates single-width atomicity for CAS failures.
Weak CAS (with sporadic failures) is tolerated in all cases
other than $\mathit{HRef}=0$.
For that case, a strong version is created by looping SC;
single-width atomicity for failures is acceptable as false negatives
are impossible unless concurrent threads modify {\tt Head} and claim the object -- i.e., $\mathit{HRef}$ is no longer $0$.

\newcommand{\algtrim}{%
\begin{algorithm2e}[H]
\Fn{\textbf{handle\_t} trim(\textbf{int} slot, \textbf{handle\_t} handle)} {
Head = Heads[slot]\tcp*{Do not alter head}

Curr = Head.HPtr\;
\If (\tcp*[f]{Non-empty list}) {Curr $\neq$ handle} {
    traverse(Curr->Next, handle);\
}
\Return Curr\tcp*{Returns a new handle}
}
\end{algorithm2e}
}

\begin{figure}
\begin{subfigure}{.5\textwidth}
\algtrim
\end{subfigure}%
\vspace{-5pt}
\caption{\hyaline{}'s trimming.}
\label{alg:trim}
\end{figure}

\section{Trimming}
\label{sec:trimming}

As mentioned in Section~\ref{sec:vectorized-smr}, the cost
of {\it enter} and {\it leave} is relatively small as long as there is no
slot contention.

Trimming provides an alternative way to avoid contention
while keeping the number of slots very small as in the simpler single-list version.
Although, \textit{trim} is
harder to use manually (we avoid it in Section~\ref{sec:eval}),
it may be practical when several operations in a row
take place on the data structure.
\textit{trim} can also be called automatically when building
\textit{quiescent-state} style memory reclamation as in~\cite{epoch2}.
Logically, \textit{trim} is equivalent to \textit{leave} followed by
\textit{enter}, but avoids the unwanted alteration of {\tt Head}.
\textit{trim} allows a thread to indicate that previously retired batches
(by any thread) are safe to delete from its perspective.
When performing multiple operations on data structures, we can call
\textit{trim} in lieu of \textit{leave} for the current operation
and \textit{enter} for the operation that follows.

The use of \textit{trim} should not be confused with a case
when \textit{enter} and \textit{leave} encapsulate several data structure operations.
In the latter case, none of the retired batches can be reclaimed
until \textit{leave} is called. In contrast, \textit{trim} dereferences
previously retired batches such that they are timely reclaimed.

In Figure~\ref{alg:trim}, we present pseudocode for trimming which extends the algorithm presented in Figures~\ref{alg:hyaline}, \ref{alg:hyaline1}, and ~\ref{alg:hyalines}.
\textit{trim} dereferences batches by traversing thread's retirement list but skips the
very first
node and \texttt{Head} entirely; it updates a per-thread \emph{handle} to reflect a shortened tail. (Note that with \textit{trim}, \hyaline{}-1 and \hyaline{}-1S can have handles which are not necessarily {\tt Null}.)
Because \textit{trim} does not update {\tt Head}, it is critical to bound
the length of retirement lists that need to be traversed,
forcing \textit{leave} and \textit{enter} afterwards.
Otherwise, if deallocation is slow, one unlucky thread
can get stuck in a state where it is traversing an ever-changing list and
deallocating more and more batches retired by other threads.

\end{document}